\documentclass[11pt,english]{article}
\usepackage{amsmath,amssymb,amsthm}
\usepackage{multicol}
\usepackage[margin=1in]{geometry}
\usepackage{graphicx,color}
\usepackage{enumitem}
\usepackage{fullpage}
\usepackage{tcolorbox}
\usepackage{babel}
\usepackage{wrapfig}
\usepackage{MnSymbol}
\usepackage{mdframed}
\usepackage{mathtools}
\usepackage{floatrow}
\usepackage{multirow}
\usepackage[ruled,linesnumbered,vlined]{algorithm2e}
\usepackage{tablefootnote}
\usepackage{thm-restate}
\usepackage{bbding}
\usepackage{pifont}
\usepackage{nicefrac}
\usepackage{xspace}
\usepackage{lineno}

\definecolor{Darkblue}{rgb}{0,0,0.4}
\definecolor{Brown}{cmyk}{0,0.61,1.,0.60}
\definecolor{Purple}{cmyk}{0.45,0.86,0,0}
\definecolor{Darkgreen}{rgb}{0.133,0.543,0.133}

\usepackage[colorlinks,linkcolor=Darkblue,filecolor=blue,citecolor=blue,urlcolor=Darkblue,pagebackref]{hyperref}
\usepackage[nameinlink]{cleveref}
\newcommand{\appref}[1]{\hyperref[#1]{Appendix~\ref*{#1}}}
\newcommand{\algline}[1]{\hyperref[#1]{line~\ref*{#1}}}

\usepackage[colorinlistoftodos,prependcaption,textsize=tiny]{todonotes}

\newif\ifdraft 
\draftfalse

\newcommand{\namedref}[2]{\hyperref[#2]{#1~\ref*{#2}}}
\newcommand{\propref}[1]{\hyperref[#1]{property~(\ref*{#1})}}

\newtheorem{theorem}{Theorem}
\newtheorem{lemma}{Lemma}
\newtheorem{definition}{Definition}
\newtheorem{claim}{Claim}
\newtheorem{observation}{Observation}
\newtheorem{corollary}{Corollary}
\newtheorem{proposition}{Proposition}

\newtheorem{remark}{Remark}

\newcommand{\localmemexp}{\delta}

\newcommand{\poly}{\mathrm{poly}}

\newcommand{\R}{\mathbb{R}}
\newcommand{\N}{\mathbb{N}}
\newcommand{\E}{\mathbb{E}}

\newcommand{\OPT}{\mathrm{OPT}}

\newcommand{\diam}{\mathrm{diam}}

\newcommand{\ddim}{\mathrm{ddim}}
\newcommand{\Exp}{\mathrm{Exp}}
\newcommand{\Prb}{\mathop{\mathrm{Pr}}}

\DeclareMathOperator{\wopen}{\mathfrak{f}}
\newcommand{\eqdef}{\coloneqq}
\DeclareMathOperator{\fl}{fl}

\newcommand{\OPTfl}{\ensuremath{\OPT^{\mathrm{fl}}}\xspace}
\newcommand{\OPTkz}{\ensuremath{\OPT^{\mathrm{cl}}_z}\xspace}

\newcommand{\cC}{\mathcal{C}}
\newcommand{\cD}{\mathcal{D}}

\newcommand{\cF}{\mathcal{F}}

\newcommand{\cI}{\mathcal{I}}

\newcommand{\cL}{\mathcal{L}}

\newcommand{\cN}{\mathcal{N}}
\newcommand{\cO}{\mathcal{O}}
\newcommand{\cP}{\mathcal{P}}

\newcommand{\Ptilde}{\tilde{\cP}_\varphi}
\newcommand{\Pin}{\tilde{\cP}_\varphi^{\rm in}}
\newcommand{\Pout}{\tilde{\cP}_\varphi^{\rm out}}
\newcommand{\phiin}{\varphi^{\rm in}}
\newcommand{\phiout}{\varphi^{\rm out}} 
\newcommand{\Oin}{{\cO}_\varphi^{\rm in}}
\newcommand{\Oout}{{\cO}_\varphi^{\rm out}}

\newcommand{\Texp}{\mathsf{TExp}}

\newcommand{\dist}{\ensuremath{d}}

\definecolor{forestgreen}{rgb}{0.13, 0.55, 0.13}

\SetCommentSty{mycommfont}

\DeclareMathOperator*{\argmax}{arg\,max}

\DeclareMathOperator{\polylog}{polylog}

\definecolor{BrickRed}{rgb}{.72,0,0}
\def\EMPH#1{\emph{\textcolor{BrickRed} {#1}}}
\newcommand{\brick}[1]{{\color{BrickRed}#1}}

\def\eps{\varepsilon}

\DeclareMathAlphabet{\mathpzc}{OT1}{pzc}{m}{it}
\newcommand{\etal}{{\em et al. }}

\newlength{\dhatheight}

\newcommand {\ignore} [1] {}

\newcommand{\initOneLiners}{%
	\setlength{\itemsep}{0pt}
	\setlength{\parsep }{0pt}
	\setlength{\topsep }{0pt}
}

\title{Geometric Hashing in MPC with Applications to Clustering and Stochastic Embeddings}

\author{%
    Sujoy Bhore\thanks{Department of Computer Science \& Engineering, Indian Institute of Technology Bombay, Mumbai, India.\\ Email: sujoy@cse.iitb.ac.in. This work is supported in part by ANRF ARG-MATRICS, Grant 002465}
    \and Arnold Filtser\thanks{Bar Ilan University, Ramat Gan, Israel. Email: arnold.filtser@biu.ac.il. This research was supported by the Israel Science Foundation (grant Nos.~1042/22 and~2021/26) and the European Research Council (ERC), grant No.~101301875 (NonMetEmbed). Funded by the European Union. Views and opinions expressed are however those of the author(s) only and do not necessarily reflect those of the European Union or the European Research Council Executive Agency. Neither the European Union nor the granting authority can be held responsible for them.}
    \and Eunjin Oh\thanks{POSTECH, Korea. Email: eunjin.oh@postech.ac.kr. This work is supported by Institute of Information \& Communications Technology Planning \& Evaluation (IITP) grant funded by the Korea government (MSIT) (No. RS-2024-00440239, Sublinear Scalable Algorithms for Large-Scale Data Analysis) and the National Research Foundation of Korea (NRF) grant funded by the Korea government (MSIT) (No. RS-2024-00358505).}
    \and Hyeonjun Shin\thanks{POSTECH, Korea. Email: hyeonjun.shin@postech.ac.kr}
    \and Himanshi Singh \thanks{Department of Computer Science \& Engineering, Indian Institute of Technology Bombay, Mumbai, India.\\ Email: himanshi@cse.iitb.ac.in}
}

\date{}
\begin{document}
\maketitle

\begin{abstract}
    In this paper we study consistent hashing (a.k.a. sparse partitions) for doubling metrics in the Massively Parallel Computation (MPC) model. 
    A $(\Gamma,\Lambda)$-hash with diameter bound $\Delta$ partitions a metric space into clusters of diameter at most $\Delta$ such that every ball of radius $\Delta/\Gamma$ intersects at most $\Lambda$ clusters in expectation. 
    Consistent hashing has emerged as a central tool in geometric MPC algorithms.
    However, so far it has been solely used in Euclidean spaces, where it can be pre-computed.
    On the other hand, in cases where the metric space is received as input and unknown in advance, no previous parallel construction was known. 
    
    We present the first fully scalable MPC algorithm over machines with local space $n^\sigma$, that given an $n$-point metric space with doubling dimension $\ddim\leq O(\sigma \log n)$, samples a randomized $(\Gamma,2^{O(\ddim/\Gamma)})$-hash in $O(\ddim/\sigma^3)$ rounds. There is a dichotomy: if the doubling dimension is $\ddim=\Omega(\sigma\cdot\log n)$, a polynomial number of rounds is required to compute such a hash.
    
    We conclude by showing that this MPC primitive yields $O(\ddim/\sigma^3)$-round fully scalable MPC algorithms for a broad class of fundamental geometric problems in doubling metrics, including stochastic tree embeddings with $O(\log n)$ expected distortion, constant factor approximation for facility location, bicriteria $(O(\eps^{-2}),1+\eps)$-approximation for $k$-median and $k$-means.
\end{abstract}

\thispagestyle{empty}

\newpage
\setcounter{secnumdepth}{5}
\setcounter{tocdepth}{2} \tableofcontents
\newpage

\section{Introduction}\label{sec:intro}
The Massively Parallel Computation (\EMPH{MPC}) model has emerged as a central theoretical framework for capturing the algorithmic challenges of modern large-scale data processing systems, including MapReduce, Hadoop, Dryad, and Spark~\cite{dean2008mapreduce, dean2004mapreduce, white2012hadoop, isard2007dryad}. Introduced by Karloff, Suri, and Vassilvitskii in their seminal work~\cite{karloff2010model}, the model has, over the past decade, become the standard abstraction for studying parallel algorithms on massive datasets; see, e.g.,~\cite{assadi2019massively, assadi2019sublinear, ghaffari2019massively, coy2022deterministic, andoni2014, goodrich2011sorting, ghaffari2019sparsifying, beame2017communication}.

An MPC system consists of $\EMPH{m}$ machines, each with local memory $\EMPH{s}$, for a total memory of $m \cdot s$. Computation proceeds in synchronous rounds: in each round, every machine performs arbitrary local computation, followed by a communication phase in which it sends and receives at most $s$ words, with each message addressed to a single destination (i.e., no broadcast). After a bounded number of rounds, the machines collectively output the solution.
The complexity of an MPC algorithm is measured primarily by the number of rounds and the local memory $s$, while the total space $m \cdot s$ is also an important parameter and typically constrained depending on the regime under consideration.
As is standard in the MPC model, local computation within each machine is assumed to be unrestricted.
A central parameter in the MPC model is the local memory size $s$, and the behavior of algorithms depends on how $s$ scales with the input size \EMPH{$N$}.
In this work, we focus on the \EMPH{fully scalable} sublinear regime, where the local memory size is $s=N^\sigma$ for $\sigma \in (0,1)$, which may be arbitrarily small, and may even depend on $N$ (though $\sigma$ is usually thought of as an absolute constant).%
\footnote{There are many works where the memory is sublinear $s=n^\sigma$ for some fixed $\sigma$. By fully scalable we mean that the algorithm works for every fixed $\sigma$, while the number of rounds may depend on $\sigma$.}
This strongly sublinear setting is particularly challenging and captures the constraints of modern large-scale systems.
This regime imposes stringent per-machine memory constraints while still allowing the total memory to be sufficient to store the input. In particular, we require that the total memory satisfies $m \cdot s \geq N$, and is ideally close to linear in $N$.
The central goal in MPC is to solve fundamental computational problems in as few rounds as possible, ideally $O(1)$.
In the rest of the introduction, we will refer to fully scalable constant-round MPC algorithms as \EMPH{FSCR-alg}.

\paragraph*{Geometric MPC.} In recent years, significant progress in MPC has been driven by the \EMPH{geometric setting}, broadly constructed as metric problems where the input admits a structured, implicit representation of distances.
Most of the previous work dealt with the case where the input consists of points in $\mathbb{R}^d$, often in high dimension, where distances can be computed locally from coordinate representations.
This additional structure enables algorithmic approaches to bypass the fundamental barriers present when the input is a graph (e.g., the $2$-vs-$1$ cycle conjecture%
\footnote{Suppose that the input is a graph $G=(V,E)$ where the edge set $E$ is distributed among the machines.
In the “$2$-vs-$1$ cycle” problem we are asked to distinguish between a graph that is a single large cycle and graphs that comprise two cycles. It is conjectured to require $\Omega(\log n)$ rounds.
There are several conditional hardness results for sublinear MPC based on the $2$-vs-$1$ cycle conjecture. See, e.g.,~\cite{NS22}.}%
).
Indeed, in recent years, there has been a large body of work that exploited this structural advantage in various memory regimes; see~\cite{andoni2014, behnezhad2022almost, goodrich2011sorting, CGJKV24, chang2024fully, CGGJ25, OS26}.

In the sublinear regime, Andoni, Nikolov, Onak, and Yaroslavtsev~\cite{andoni2014} initiated the study of geometric graph problems.
They obtained an FSCR-alg $(1+\eps)$-approximation for MST for metrics with constant doubling dimension (and a $(1+\eps)$-approximation for Earth Mover Distance in $\mathbb{R}^2$).
A lot of attention was devoted to fundamental clustering problems ($k$-median/means/center, facility location).
Most of the previous work~\cite{bhaskara2018distributed,chen2016communication,cohen2021parallel,coy2023parallel, enrico2023, ceccarello2017mapreduce, ceccarello2024fast} was based on classical techniques (e.g., core-sets) and typically requires local memory $s \ge \Omega(k)$, thus not yielding fully scalable MPC algorithms.

Recently, the situation drastically changed via the introduction of consistent hashing.
Czumaj, Gao, Jiang, Krauthgamer, and Vesel{\'{y}}~\cite{CGJKV24} initiated the use of consistent hashing in MPC and used it to design an FSCR-alg constant-approximation for facility location for points in high-dimensional $\mathbb{R}^d$.
This was the first FSCR-alg for facility location.
Further,~\cite{CGJKV24} used consistent hashing to design an FSCR-alg bicriteria $(O(\eps^{-2}),1+\eps)$-approximation for $k$-median and $k$-means for points in high-dimensional $\mathbb{R}^d$ (previous FSCR-algs had polylogarithmic bicriteria approximation~\cite{bhaskara2018distributed}).
In a subsequent work, Czumaj \etal~\cite{CGGJ25} utilized consistent hashing to construct an FSCR-alg $O(\frac{\log n}{\log\log n})$-approximation for $k$-center for points in high-dimensional $\mathbb{R}^d$—the first $k$-center FSCR-alg for points in high dimension.\footnote{\cite{CGGJ25} also provided improved $(2+\eps)$-approximation for points in low-dimensional $\mathbb{R}^d$.}
Lastly, Goranci \etal~\cite{goranci2026tree} constructed FSCR-alg stochastic tree embeddings for points in high-dimensional $\mathbb{R}^d$ with optimal expected distortion $O(\log n)$. Previous FSCR-algs for stochastic tree embeddings for points in $\mathbb{R}^d$ had expected distortion $\tilde{O}(\log^{1.5}n)$~\cite{ahanchi2023massively}.

\begin{wrapfigure}{r}{0.3\textwidth}
	\begin{center}
		\vspace{-30pt}		\includegraphics[width=0.95\textwidth]{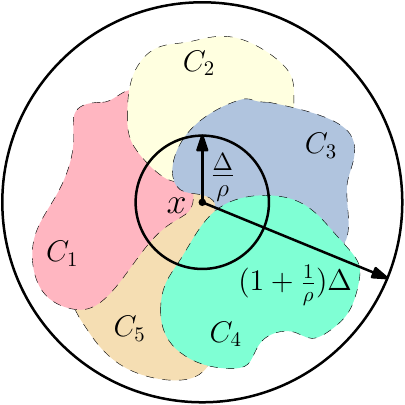}
		\vspace{-15pt}
	\end{center}
	\vspace{-10pt}
\end{wrapfigure}
\paragraph*{Consistent Hashing (a.k.a Sparse Partitions).} Given a metric $(X,d_X)$, a \EMPH{$(\rho, \tau, \Delta)$-sparse partition} is a partition of $X$ into clusters of diameter $\le \Delta$ such that every ball of radius $\frac{\Delta}{\rho}$ intersects at most $\tau$ clusters.
In the example on the right, the ball $B_X(x,\frac{\Delta}{\rho})$ intersects $5$ different clusters: $C_1,\dots,C_5$.
Sparse partitions are useful, especially in sublinear regimes, as they allow us to efficiently approximate balls and manipulate them in various ways. Indeed, let $\cC_x$ be the collection of clusters intersecting $B_X(x,\frac{\Delta}{\rho})$.
It holds that $B_X(x,\frac{\Delta}{\rho}) \subseteq \cup \cC_x \subseteq B_X\left(x,(1+\frac{1}{\rho})\Delta\right)$.
Maintaining and manipulating such clusters is often much easier than similar operations on balls.

Sparse partitions were introduced by Jia \etal~\cite{JLNRS05} (see also~\cite{awerbuch1990sparse}), who used them to construct universal TSP and universal Steiner tree (see~\cite{BDRRS12,BCFHHR24} for follow-ups).
Later, similar notions of scattering and shortcut partitions~\cite{Fil24scattering,CCLMST23Planar,CCLMST24} were used to resolve the Steiner point removal problem and to construct sparse tree covers.
Existential optimal constructions are now known in a variety of settings, including general metrics~\cite{JLNRS05}, doubling (and Euclidean) metrics~\cite{Fil24scattering}, and minor-free graphs~\cite{CF25} (see \Cref{sec:related_work} for further details).
Czumaj \etal~\cite{CJFKVY23} coined the term \EMPH{(geometric) consistent hashing} to emphasize the computational aspects of these constructions.
Beyond FSCR-algs for MPC (mentioned above), consistent hashing for high-dimensional $\mathbb{R}^d$ had been used for $O(\frac{d}{\log d})$-approximation in dynamic streaming algorithms (single pass, polylog space) for facility location,
for fully dynamic stochastic tree embedding with distortion $O_\eps(\log n)$ and update time $O(n^\eps+d)$~\cite{goranci2026tree},
for fully dynamic $\poly(\frac{1}{\eps})$-approximation algorithms for $k$-means with $\tilde{O}(k^\eps)$ update time and polylogarithmic recourse,
and for almost-linear-time constant-approximation for $k$-center~\cite{FJLNPY25}.

Interestingly, while sequential constructions of sparse partitions are known for various graph families and metric spaces, all this recent success story of consistent hashing for MPC FSCR-algs, streaming, fully dynamic, and almost-linear-time algorithms is exclusive to Euclidean space.
This is somewhat natural: indeed, Euclidean space is fixed. One can leisurely compute a consistent hash offline and then hard-code it into the local algorithm of each machine.
On the other hand, in most other scenarios, the metric space arrives as input, and it is impossible to precompute the consistent hash in advance.
But is the situation indeed so dire for other metric spaces?
Is it the case that only Euclidean space (or other predetermined spaces) is destined to enjoy the fruits of consistent hashing?

This leads to the following fundamental question, which is the driving core behind our paper:

\begin{quote}
\emph{Can one design consistent hashing schemes in the fully scalable MPC model, in constant rounds, for broader metric spaces, such as doubling or general metrics?}
\end{quote}

\paragraph*{Distance-Label Model.}
Prior works~\cite{andoni2014, CGJKV24, CGGJ25, ahanchi2023massively, OS26} on metric problems typically assume, often implicitly, access to a ground set $S$ endowed with a distance function $\delta(\cdot,\cdot)$. We make this assumption explicit via the \EMPH{Distance-Label MPC} model.
In this model, we are given a set of points $P$, where each point $p_i \in P$ is associated with a label $\ell_i$ of size $O(\log n)$ bits. The distance between any two points $p_i$ and $p_j$ is not given explicitly; instead, it can be computed in $O(1)$ time from their labels $\ell_i$ and $\ell_j$.

\subsection{Our contributions}
\paragraph*{Main theorem.}
In this work, our primary goal is to obtain a fully scalable MPC (\EMPH{FSCR-alg}) algorithm for consistent hashing schemes for a broader class of metric spaces. To this end, we develop our approach in doubling metrics and establish the following upper bound.

\begin{theorem}[Informal; see \Cref{thm:main_doubling}]\label{thm:mainInformal}
For any metric space of doubling dimension $\ddim$, there is a fully scalable MPC algorithm that, given $n$ points distributed among machines with local memory $n^\sigma$ (where $\ddim\le O(\sigma\cdot\log n)$), constructs a $(\Gamma,2^{O(\ddim/\Gamma)},\Delta)$-consistent hashing (sparse partition) for arbitrary $\Gamma\ge 16$ and $\Delta$. The algorithm runs in $O(\ddim/\sigma^3)$ rounds if the input point set has a polynomially bounded aspect ratio.
\end{theorem}

The padding-sparsity tradeoff of our consistent hashing is optimal~\cite{Fil24scattering}.
Note that this tradeoff is also optimal for $d$-dimensional Euclidean space.
Furthermore, in cases where the intrinsic doubling dimension $\ddim$ of a given point set $P\subset\R^d$ is considerably smaller than the actual Euclidean dimension $d$, it will be much preferable to use \Cref{thm:mainInformal} over the existing consistent hashes of high-dimensional Euclidean spaces.
Algorithms for Euclidean point sets with small intrinsic doubling dimension are extensively studied \cite{bartal2014impossibility, choudhary2014local, indyk2007nearest, aiger2014reporting,anagnostopoulos2018randomized, ding2019geometric, ding2020metric}.

At a high level, we revisit the classical MPX-based sparse partition framework.
However, in order to implement it in MPC, we need to take the intersection of $\poly(\frac1\sigma)$ partitions in different scales.
The classic construction \cite{Fil24scattering}, while providing worst case guarantees, bounded the number of intersections only by $\tilde{\Omega}(\ddim)$, regardless of how large the padding parameter $\Gamma$ is. 
Implementing this in our multi-scale approach would have blown up the number of intersections to exponential (even for huge $\Gamma$).
We provide a more careful analysis that bounds the expected number of intersections by $2^{O(\ddim/\Gamma)}$, which goes to $1$ as $\Gamma$ grows larger.
Thus eventually we don't lose anything by our multi-scale approach.
See \Cref{rem:ComparisonFil24} and \appref{apx:comparison-fil24} for a detailed comparison with \cite{Fil24scattering}.

\paragraph*{Lower bound.}
As general $n$-point metrics admit a $(O(\log n),O(\log n))$-sparse partition scheme \cite{JLNRS05}, one might hope to remove the bounded doubling dimension assumption from \Cref{thm:mainInformal}, and construct consistent hashing for general metrics in MPC.
Unfortunately, this is impossible.
In \Cref{thm:mainInformal} we require that $\ddim\le O(\sigma\cdot\log n)$. In other words, the local space has to be at least $s=n^\sigma=2^{\Omega(\ddim)}$.
General metrics might have doubling dimension as large as $\log n$.
We show that if this assumption does not hold, then a polynomial number of rounds is required to compute a consistent hash. 
Thus there is a sharp dichotomy: for $\ddim=O(\sigma\cdot\log n)$ there is an efficient $O(\ddim/\sigma^3)$-round algorithm, while for $\ddim=\Omega(\sigma\cdot\log n)$, one needs either a polynomial number of rounds, or $n^{1-\sigma}\cdot2^{\ddim}$ total space ($n^{2-\sigma}$ for general metrics).

\begin{theorem}[Informal; see \Cref{thm:LB-hashing}]
Any fully scalable MPC algorithm that constructs consistent hashing schemes in doubling metrics using $O(n^{1+\varepsilon})$ total space and $n^\sigma$ local memory must take at least $\Omega\!\left(2^{\ddim}/n^{\varepsilon+\sigma}\right)$ rounds.
\end{theorem}

\paragraph*{Applications.}
We now turn to applications of consistent hashing. Our consistent hash bears fruit for several fundamental problems (mostly in a black-box fashion), such as clustering, stochastic tree embeddings, and EMD. 
These problems have been extensively studied in the MPC model. 
However, existing approaches~\cite{ceccare2019,sungjin2015,LIANG2025114975,  ceccarello2017mapreduce, MIRJALALI2024114296,ceccarello2017,mazzetto2019,enrico2023} require large local memory (typically $\tilde{O}(n^{1/2})$ per machine), incur additional dependence on parameters such as $k$ (number of clusters), and thus fall short of achieving the ``holy grail'' fully scalable regime. 
Consequently, fully scalable $O(1)$-round MPC algorithms for a broad class of fundamental problems remain out of reach beyond Euclidean settings.

\medskip
In the following description of our applications we assume that the doubling dimension $\ddim$ and $\sigma$ are absolute constants. 
Nevertheless, our results hold even for moderately large doubling dimension, as long as $\ddim\le O(\sigma\cdot\log n)$. Then the number of rounds is $O(\ddim/\sigma^3)$.

\medskip
\noindent \textbf{Stochastic Embeddings:}
Our framework yields a fully scalable MPC algorithm for stochastic tree embeddings for doubling metrics, achieving optimal $O(\log n)$ expected distortion in constant rounds and near-linear total space; see \Cref{cor:HST}. 
Stochastic tree embeddings are a very powerful tool, as they allow us to treat various metrics as if they were a tree. 
Indeed, we use our stochastic tree embedding to obtain the first FSCR-alg for approximating EMD for doubling metrics (see more details below). 

We also show that stochastic tree embeddings are impossible for general metrics in MPC. Indeed, $\Omega(2^{\ddim}/n^{\varepsilon+\sigma})$ rounds are required (for local space $n^\sigma$, and total space $n^{1+\eps}$); see \Cref{thm:LB-hst}.

\medskip
\noindent \textbf{Facility Location:}
We obtain a fully scalable, constant-round MPC constant-factor approximation algorithm for facility location in doubling metrics, with near-linear total space; see \Cref{cor:facilityLocation}. 
No fully-scalable approximation for doubling metrics was previously known.

\medskip
\noindent \textbf{$k$-Median and $k$-Means:}
We obtain fully scalable MPC algorithms for $k$-median and $k$-means in doubling metrics (in fact for a more general $(k,z)$-objective), achieving a $(1+\varepsilon, O(\varepsilon^{-2}))$-bicriteria approximation in constant rounds, while using near-linear total space; see \Cref{cor:kmedian}.
No fully-scalable approximation for doubling metrics was previously known.

\medskip
\noindent \textbf{Earth Mover's Distance (EMD):}
We obtain a fully scalable, constant-round MPC, $O(\log n)$-approximation algorithm for EMD in doubling metrics; see \Cref{cor:emd}. 
No fully-scalable approximation for doubling metrics was previously known.

\medskip
\noindent \textbf{$r$-Gather (a.k.a. minimum size clustering):}
We obtain a fully scalable, $O((\log\log n)^2)$-round MPC, $O(\log\log\log n)$-approximation algorithm for the $r$-gather problem in doubling metrics; see \Cref{cor:r-gather} and \Cref{cor:r-gather-nice}. 
No fully-scalable approximation for doubling metrics was previously known.

\subsection{Related work}\label{sec:related_work}
Given a metric space $(X,d_X)$, a $(\rho,\tau,\Delta)$-\EMPH{sparse cover} is a collection of clusters $\cP$ with diameter at most $\Delta$ such that every ball of radius $\Delta/\rho$ is fully contained in some cluster, and every vertex belongs to at most $\tau$ clusters.
If  $(X,d_X)$ admits a $(\rho,\tau,\Delta)$-sparse cover for every $\Delta>0$, we can simply say that it admits a $(\rho,\tau)$-sparse cover.
It is known that sparse partitions (a.k.a. consistent hashing) and sparse covers are essentially equivalent (see \cite{Fil24scattering}).
Sparse covers were introduced by Awerbuch and Peleg \cite{AP90} for general graphs, who showed that they admit $(4k-2,O(k\cdot n^{\frac1k}))$-sparse covers for every integer $k\geq 1$.
Metric spaces with doubling dimension $\ddim$ admit $(\ddim,\tilde{O}(\ddim))$-sparse covers, or $\left(O(t),2^{\frac{\ddim}{t}\cdot\ddim\cdot\log t}\right)$-sparse covers for every $t\ge1$.
Sparse covers were extensively studied for $K_r$-minor free graphs \cite{KPR93,FT03,KLMN04,AGMW10,BLT14,Fil19padded,FFIKLMZ24,Filtser25}, where recently Conroy and Filtser \cite{CF25} showed that $K_r$-minor free graphs admit $(O(1),O(r^4))$-sparse covers.

Given a weighted graph $G=(V,E,w)$, the \EMPH{strong diameter} of a cluster $C\subseteq V$ is its maximum induced distance, i.e., $\max_{u,v\in C}d_{G[C]}(u,v)$. In this paper we only use weak diameter, which is defined as $\max_{u,v\in C}d_{G}(u,v)$, the maximum distance w.r.t. the original graph distances. It is considerably harder to construct partitions with strong diameter guarantees.
In particular, it is impossible to construct strong $(O(1),O(1))$-sparse partitions even for trees \cite{Fil24scattering}.
$(O(1),O(1))$-sparse partitions are known only for graphs with either constant doubling dimension, or constant pathwidth \cite{Fil24scattering}.

\EMPH{Scattering partitions} are similar to sparse partitions, where the guarantee is that every shortest path intersects a small number of clusters (instead of a ball), and the diameter guarantee is strong.
Other than bounded-doubling and pathwidth graphs, chordal \cite{Fil24scattering} and series-parallel graphs \cite{FL22} are also known to admit $(O(1),O(1))$-scattering partitions. It is still open whether planar graphs admit $(O(1),O(1))$-scattering partitions.
\EMPH{Shortcut partitions} are similar to scattering partitions, where the difference is (roughly) that instead of the guarantee that every shortest path intersects a small number of clusters, the guarantee is that between every pair of vertices there is some approximate shortest path that intersects $O(1)$ clusters. 
It is known that planar graphs, or more generally $K_r$-minor free graphs, admit $(O(1),O(1))$-shortcut partitions \cite{CCLMST23Planar,CCLMST24}. 
Scattering and shortcut partitions have been used to obtain $O(1)$-distortion solutions for the Steiner point removal problem \cite{Fil19SPR,Fil24scattering,CCLMST24}.

Consistent hashing in Euclidean $\R^d$ was extensively studied. For every $\beta>1$, it is known how to construct a $(\beta,\tilde{O}(d)\cdot e^{O(d/\beta)})$-consistent hashing \cite{Fil24scattering}, which is tight up to second-order terms \cite{Fil24scattering}.
Note that our \Cref{thm:main_doubling} removes the $\tilde{O}(d)$ factor from the sparsity (due to the fact that our guarantees are only in expectation). \cite{BCFJJL25} constructed a different consistent hash function with the same parameters that can be stored and computed using $\tilde{O}(d^2)$ space. Unfortunately, evaluating both near-optimal hash functions \cite{Fil24scattering,BCFJJL25} takes $\exp(d)$ time.
On the other hand, it is known how to construct a $(\beta,\poly(d)\cdot e^{O(d/\beta^{\frac23})})$-consistent hash \cite{FJLNPY25,BCFJJL25} that can be both stored and evaluated in $\poly(d)$-time. Constructing efficient consistent hashing for $\R^d$ with optimal parameters is a fascinating open question.

\section{Preliminaries}
In this section, we introduce basic notions used in our main results.
We consider a metric space $(X,\dist)$ of doubling dimension $\ddim$. 
Throughout this paper, we assume that the aspect ratio of $X$ is polynomially bounded.
The \EMPH{doubling dimension} $\ddim$ of a metric space is the smallest constant such that any ball of radius $r$, $r \geq 0$, can be covered by at most $2^{\ddim}$ balls of radius $r/2$~\cite{GKL03}. It is well-known that a $d$-dimensional Euclidean space has doubling dimension $\Theta(d)$. Doubling metrics admit several useful structural properties, including nets and the packing property.
The \EMPH{packing property} can be stated as follows: 
If $S \subseteq X$ has minimum interpoint distance~$\delta$ and is contained in a ball of radius~$R$, then $|S| \le (2R/\delta)^{\ddim}$.
For $\delta \geq 0, \delta' \geq 0$, a set $N \subseteq X$ is called a \EMPH{$(\delta, \delta')$-net} if for any point $x$ of~$X$ there is a point $q \in N$ such that $d_X(x, q) \leq \delta$, and for any pair of net points $(q, q')$, it holds that $d_X(q, q') \geq \delta'$. We call $N$ a \EMPH{$\delta$-net} if $\delta = \delta'$.
The packing property implies that any $\delta$-net contained in a ball of radius $R$ has size at most $(2R/\delta)^\ddim$.

For a point set $S\subseteq X$, we use \EMPH{$\diam(S)$} to denote the diameter of $S$, that is, $\diam(S)=\max_{x,y\in S} \dist(x,y)$.
For a point $v\in X$ and a value $r>0$, we let \EMPH{$B_X(v,r)$} denote the ball centered at $v$ with radius $r$.

\paragraph*{Consistent hashing.}
We define (geometric) consistent hashing as follows. A hash $\varphi$ induces a partition of $X$ into \emph{clusters} of diameter at most $\Delta$ such that for any two points $x,y\in X$, they are contained in the same cluster if and only if $\varphi(x)=\varphi(y)$.
A partition induced by a $(\Gamma,\Lambda)$-hash with diameter bound $\Delta$ is called a $(\Gamma,\Lambda,\Delta)$-sparse partition. 

\begin{definition}[Consistent Hashing]\label{def:consistent-hashing}
Let $(X,\dist)$ be a metric space, $\cI$ an index set, and $\Delta>0$ a scale.
A distribution $\mathcal{H}$ over mappings $\varphi : X \to \cI$ is called a
randomized $(\Gamma,\Lambda)$-hash with diameter bound $\Delta$ if:
\begin{enumerate}[itemsep=-0.1em]
\item (Diameter) For every $\varphi$ in the support of $\mathcal{H}$ and
every $i\in \cI$,
$\diam(\varphi^{-1}(i)) \le \Delta$.
\item (Expected consistency) For every $S\subseteq X$ with
$\diam(S)\le \Delta/\Gamma$, we have 
$\mathbb{E}_{\varphi\sim\mathcal{H}}[|\varphi(S)|] \le \Lambda$.
\end{enumerate}
Moreover, if $\mathcal H$ is supported on a single mapping $\varphi$,
we call it a \emph{deterministic} $(\Gamma,\Lambda)$-hash.
\end{definition}

In Euclidean spaces, consistent hashing is input-independent and can be defined implicitly, rather than constructed algorithmically. 
In particular, a shared random seed suffices for all machines to obtain the same hash, and the random seed can be generated by a single machine. 
However, in doubling metrics, the situation is inherently different: consistent hashing must be computed from the input and thus requires communication between machines. Moreover, even if a machine stores the hash values of the points stored in the machine, it is not straightforward to determine the set of hash values $\varphi(B_X(x,\Delta/\Gamma))$ for a given point $x\in X$.
Due to this difference, we need to introduce the notion of \EMPH{neighborhood hash pairs} for doubling metrics explicitly, defined as follows.
Note that if we sample a $(\Gamma,\Lambda)$-hash $\varphi$ then the expected size of the set of neighborhood hash pairs is $n\cdot \Lambda$.

\begin{definition}[Neighborhood hash pairs]\label{def:hashpairs}
Given a $(\Gamma,\Lambda)$-hash $\varphi$ with diameter bound $\Delta$ for a metric space $(X,d_X)$, the \EMPH{set of neighborhood hash pairs} is 
    $$\cP_\varphi=\left\{(x,q)\mid x\in X,q\in\cI,\mbox{ and there is }y\in B_X(x,\Delta/\Gamma)\mbox{ s.t. }\varphi(y)=q\right\}~,$$
the collection of point-hash index pairs $(x,q)$ such that there is some point $y\in X$ at distance at most $d_X(x,y)\le\Delta/\Gamma$ from $x$ that is mapped to $q$.\\
A subset $\tilde{\cP}_\varphi\subseteq X\times\cI$ is called a \EMPH{set of approximate neighborhood hash pairs} if 
$\tilde{\cP}_\varphi\supseteq \cP_\varphi$, and in addition, for every $(x,q)\in \tilde{\cP}_\varphi$, it holds that $\varphi^{-1}(q)\subseteq B_X(x,2\Delta)$.
\end{definition}

\paragraph*{MPX clustering and exponential distribution.}
We use the clustering algorithm of Miller, Peng, and Xu~\cite{MPX13}, as generalized in~\cite{Fil24scattering}. Given a set of centers $N \subseteq X$ with shifts $\{\delta_t\}_{t \in N}$, each vertex~$v$ is assigned to the cluster of the center $t \in N$ that maximizes $f_v(t) = \delta_t - d_X(t,v)$. 
The following property is standard~\cite{MPX13, Fil24scattering}.

\begin{proposition}[Ball Intersection]\label{prop:ball}
Let $t^*$ be a center that maximizes $f_v$ among all centers. If $f_v(t') < f_v(t^*)- 2r$, then $B_X(v,r)$ does not intersect the cluster of~$t'$.
\end{proposition}

In the original paper~\cite{MPX13}, the shifts are drawn from an exponential distribution. 
Let \EMPH{$\Exp(\lambda)$} denote the exponential distribution with rate~$\lambda$ and density $g(x) = \lambda e^{-\lambda x}$ for $x \ge 0$. A key property is \emph{memorylessness}: if $Y \sim \Exp(\lambda)$, then for all $a, b \ge 0$, $\Prb[Y \ge a + b \mid Y \ge a] = \Prb[Y \ge b]$.
In this case, since 
a value drawn from $\Exp(\lambda)$ can be arbitrarily large, 
the diameter of clusters obtained by the MPX algorithm can be unbounded in the worst case.
To avoid this, we use \EMPH{truncated exponential distribution} denoted by 
\EMPH{$\Texp(\lambda,\lambda_\top)$}, that is, an exponential distribution with rate parameter $\lambda$ conditioned on the event that the outcome is at most $\lambda_\top$.
The density function is then
	$g(y) = \frac{ \lambda\cdot e^{-\lambda\cdot y} }{1 - e^{-\lambda \lambda_\top}}$.
While $\Texp(\lambda,\lambda_\top)$ has bounded support (at most $\lambda_\top$), it is no longer memoryless.

\paragraph*{MPC primitives.} We use the following standard MPC primitives in our algorithms.

\begin{lemma}[Sorting in MPC; {~\cite{goodrich2011}}]\label{lem:sort}
Given $n$ comparable items
distributed across $m = n/s$ machines with space $s = n^\sigma$ per machine, the items can be sorted and redistributed so that each machine holds a contiguous block of $s$ items in the sorted order. The algorithm runs in $O(\log_s n)$ MPC rounds with high probability, using $O(n)$ total space.
\end{lemma}

\begin{lemma}[Broadcast in MPC; {~\cite{dinitz2019massively}}]\label{lem:broadcast}
    Given two machine IDs $a,b$, a machine can send, i.e., broadcast, a message of size at most $s=n^\localmemexp$ to all machines $M_a,M_{a+1},\ldots,M_b$ in $O(\log_sn)$ rounds, where $M_i$ denotes the machine with ID $i$.
\end{lemma}

\begin{lemma}[Net Construction in MPC; adapted from Lemma~6.3 of~\cite{andoni2014}]
\label{lem:net-mpc}
Let $(X, d_X)$ be an $n$-point metric space with doubling dimension $\ddim$, which can be covered by a ball of radius~$R$. 
Suppose that the points of $X$ are distributed among machines with local memory $s\ge C\log^3(n/\gamma)$, for a sufficiently large constant $C$. 
Let $t = \lfloor\frac{\log s}{3 \cdot \ddim}\rfloor$ and $\delta = 8R/2^t$. 
Then there is an MPC algorithm that with probability $1-\gamma$ computes a $(\delta, \delta/4)$-net of size at most $\sqrt[3]{s}$. The following properties hold:
\begin{enumerate}
\item The required number of machines is $O(n/s)$. Thus the total space is $O(n)$.
\item The number of rounds is 
$O(\log^2_s n)$.
\item The computation time per machine per round is
    $2^{O(\ddim)} \cdot s \cdot \log s$.
\end{enumerate}
\end{lemma}

The net scale $\delta = 8R/2^t$ is determined by the space-per-machine ratio $n/m = s$. Given a cluster of diameter~$D$, it is contained in a ball of radius~$D$ around any point in the cluster, and thus the finest $\delta$-net that fits in memory has $\delta = 8D/2^t$, where $t = \lfloor\frac{\log s}{3 \cdot \ddim}\rfloor$. Note that \Cref{lem:net-mpc} actually produces a $(\delta, \delta/4)$-net, which is in particular a $\delta$-net (every point is within distance~$\delta$ of a net point) whose minimum interpoint distance is at least~$\delta/4$.

\section{Consistent hashing for doubling metrics}\label{sec:doubling}
In this section we show how to construct a \EMPH{$(\Gamma, \Lambda)$-consistent hash} with diameter bound $\Delta$ (see \Cref{def:consistent-hashing}) for a metric space $(X, d_X)$ with doubling dimension $\ddim$, working in the fully scalable MPC model with $s = n^{\sigma}$ space per machine. Our construction achieves $\Lambda = 2^{O(\ddim/\Gamma)}$ in expectation and terminates in $O(\ddim/\sigma^{3})$ MPC rounds using $\widetilde{O}(n)$ total space with high probability.

Filtser~\cite{Fil24scattering} constructed \EMPH{sparse partitions} for doubling metrics based on the classical clustering algorithm of Miller, Peng, and Xu~\cite{MPX13} (MPX) with \emph{``betailed exponential shifts"}\footnote{A \emph{betailed} exponential distribution with parameters $\lambda,\lambda_{\top}$ is an exponential distribution with parameter $\lambda$, where a value above $\lambda_{\top}$ collapses to  $\lambda_{\top}$}.
In this section, we revisit this construction. However, in contrast to the worst-case analysis of~\cite{Fil24scattering}, we provide an alternative analysis based on the expected number of intersections. This improves the bound on the number of intersections. Building on this analysis, we develop a parallel algorithm that implements this construction in the MPC model.

The section is organised as follows. We first describe the MPX-based partition of~\cite{Fil24scattering} for a single point set (\Cref{sec:sequential}) and establish an expected bound on the number of cluster intersections per ball. We then present a \emph{parallel construction} (\Cref{sec:parallel}) that applies this partition across all clusters simultaneously in each MPC round, iterating until cluster diameters reach the target~$\Delta$. 
Then we prove that using the sparse partition constructed by \Cref{alg:par}, we can construct consistent hash functions for metric spaces with bounded doubling dimension in MPC in $O(\ddim/\sigma^{3})$ rounds and $O(n \log n)$ total space. We also provide a construction for \EMPH{approximate neighborhood hash pairs} (see \Cref{def:hashpairs}).
For many applications including facility location, $k$-median and $k$-means clustering and minimum-size clustering, it suffices to use approximate neighborhood hash pairs. However, for stochastic tree embedding, we need a representation of a hash stronger than approximate neighborhood hash pairs.
We define this stronger representation and show how to obtain the representation in Section~\ref{sec:neighborhood-oracle}.

\subsection{Sparse partition for doubling metrics}\label{sec:sequential}

First, we describe the sparse partition construction of~\cite{Fil24scattering}, which uses the MPX clustering framework of Miller, Peng, and Xu~\cite{MPX13} for metrics of bounded doubling dimension~$\ddim$. 
Given a finite metric space $(X, d_X)$ of doubling dimension~$\ddim$ and diameter~$D$, together with a $(\Delta,\Delta/4)$-net $N$ of~$X$ (for $\Delta \le D$), we partition $X$ into clusters of diameter at most~$3\Delta$. 
We assign to each net point $q\in N$ a random shift \EMPH{$\delta_q$} drawn from a truncated exponential distribution with rate \EMPH{$\lambda$}$=c \cdot \ddim/\Delta$ 
(for a constant $c \ge 8$ to be determined later), truncated at \EMPH{$\lambda_\top$}$=\Delta/2$. 
Each point $x \in X$ is then assigned to the center $q$ that maximizes \EMPH{$f_x(q)$}$ = \delta_q - d_X(x,q)$, i.e., the center whose \emph{``spreading wave''}\footnote{ We can think of each center $q$ waking up at time $-\delta_q$ and spreading its wave, i.e., other points begin joining its cluster. See \cite{Fil24scattering} for further explanation and intuition.} reaches~$x$ first.  
The truncation enforces a diameter bound, while the exponential shifts ensure that nearby centers are unlikely to compete for the same region of space. 

\begin{lemma}[Sparse Partition for Doubling Metrics]\label{lem:seq} Let $(X, d_X)$ be a finite metric space with doubling dimension~$\ddim$ and diameter~$D$, and let $N$ be a $(\Delta,\Delta/4)$-net of~$X$ with $\Delta \le D$. We can construct a partition $\Pi$ such that
\begin{enumerate}
	\item[(i)] All clusters of $\Pi$ have a diameter at most~$3\Delta$.
	\item[(ii)] For any ball $B_x = B_X(x,r)$ (for $r\le\frac\Delta2$), the
	expected number of intersections between $\Pi$ clusters and $B_x$ is at most
	$e^{O(\ddim\cdot\frac{r}{\Delta})}$. In particular, for $r=O(\frac{\Delta}{\ddim})$, there are at most $1+O(\ddim\cdot\frac{r}{\Delta})$
	expected intersections.
\end{enumerate}
\end{lemma}

\begin{proof} 

Let $N$ be a $(\Delta,\Delta/4)$-net. A point $x$ joins the cluster of the center $q\in N$ that maximizes $f_x(q) = \delta_q - d_X(x,q)$.
All shifts satisfy $0 \le \delta_q \le \Delta/2$ (by truncation). For any point $x \in X$, there is a net point $q_x \in N$ such that $d_X(x,q_x) \le \Delta$. Hence for $x$ and $q_x$, $f_x(q_x) = \delta_{q_x} - d_X(x,q_x) \ge -\Delta$. Since $x$ joins center $q$, $f_x(q) \ge f_x(q_x) \ge -\Delta$. This gives us $d_X(x,q) \le \delta_q + \Delta \le \frac32\Delta$. 
By the triangle inequality, the cluster centered at $q$ has diameter at most $3\Delta$.
This proves guarantee (i).
Furthermore, any center $q'$ whose partition intersects $B_X=B_X(x,r)$ at a point $z\in B_X$ must satisfy $d_{X}(x,q')\le d_{X}(x,z)+d_{X}(z,q')\le r+\frac{3}{2}\Delta\le2\Delta$. Define $\brick{N_x} = N \cap B_X(x, 2\Delta)$. 
Note that only centers from $N_x$ have a nonzero chance to intersect $B_x$.
By the packing property, since $N_x$ is a $(\Delta,\Delta/4)$-net with minimum interpoint distance~$>\Delta/4$, contained in a ball of radius~$2\Delta$:
\begin{equation}\label{eq:Nx}
	|N_{x}|\le\left(\frac{2\cdot2\Delta}{\Delta/4}\right)^{\ddim}=16^{\ddim}~.
\end{equation}

Now we analyze the number of intersections. Denote by $\Lambda$ the number of clusters intersecting~$B_x$ minus~$1$. That is, if the entire ball~$B_x$ is contained in a single cluster, then $\Lambda = 0$; in general the number of intersections is $\Lambda + 1$.
In the rest of the analysis we consider only the centers in $N_x$, as the clusters centered in all the other centers will not intersect $B_x$.
First, suppose that the shifts of the centers in $N_x$ were sampled using a classical (memoryless) exponential distribution $\Exp(\lambda)$.
We will analyze $\Lambda$ based on this assumption. 
Later, we will incorporate truncation.

Order the centers in $N_x$ w.r.t. their $f_x$ values. 
That is $N_x=\left\{q_{(1)},q_{(2)},\dots,q_{(i)},\dots\right\}$ where $f_x(q_{(1)})\ge f_x(q_{(2)})\ge\cdots\ge f_x(q_{(i)})\ge\cdots$.
Note that $q_{(1)}$ is a random variable.
Recall that by \Cref{prop:ball} a center $q_{(i)}$ can intersect $B_x$ only if $f_x(q_{(i)})\ge f_x(q_{(1)})-2r$.
It follows that for at least $k+1$ intersections to occur, it must be the case that $f_x(q_{(k+1)})\ge f_x(q_{(1)})-2r$.
We think of the process of sampling the shifts as a gradual process, where we grow the shifts continuously, and occasionally some centers ``die out''.
Consider the point in time where there are exactly $k$ active centers remaining. At this point the identity of $\{q_{(k+1)},q_{(k+2)},\dots\}$ is already determined, while $\{q_{(1)},\dots,q_{(k)}\}$ are yet undetermined.
Denote by $\cN\subseteq N_x$ this set of yet active centers.
For the event $\Lambda\ge k$, it must happen that every active center $t\in\cN$ will grow additionally by at most $2r$.
Fix $t\in\cN$. By memorylessness, conditioned on $t$ still being active at this point, the additional amount by which it grows is distributed as an independent exponential random variable. 
The probability that it will grow by at most $2r$ is

\begin{align*}
	& \Pr\left[f_{x}(t)-f_{x}(q_{(k+1)})\le2r\mid f_{x}(t)\ge f_{x}(q_{(k+1)})\right]\\
	& \quad=\Pr\left[\delta_{t}\le2r+\delta_{q_{(k+1)}}-d_{X}(x,q_{(k+1)})+d_{X}(x,t)\mid\delta_{t}\ge\delta_{q_{(k+1)}}-d_{X}(x,q_{(k+1)})+d_{X}(x,t)\right]\\
	& \quad\le\Pr[\Exp(\lambda)\le2r]=1-e^{-2r\lambda}~.
\end{align*}
As the samples of the shifts are completely independent, we conclude
\begin{align}
	\Pr[\Lambda\ge k] & \le\Pr[f_{x}(q_{(1)})-f_{x}(q_{(k+1)})\le2r]\nonumber \\
	& =\Pr[\forall t\in\mathcal{N},~f_{x}(t)-f_{x}(q_{(k+1)})\le2r]\le\left(1-e^{-2r\lambda}\right)^{k}~.\label{eq:tail}
\end{align}
It follows that,
\begin{equation}
	\E[\Lambda]=\sum_{k\ge1}\Prb[\Lambda\ge k]\le\sum_{k\ge1}\left(1-e^{-2r\lambda}\right)^{k}=\frac{1-e^{-2r\lambda}}{e^{-2r\lambda}}=(1-e^{-2r\lambda})\cdot e^{2r\lambda}~.\label{eq:Ef-before}
\end{equation}

Next, we incorporate the truncation. 
Let $\Psi$ denote the event that some center $q \in N_x$ sampled a shift ${\delta}_q > \Delta/2$.
Equivalently, the truncated shifts used by the algorithm are obtained by conditioning the auxiliary exponential shifts on $\overline{\Psi}$.

By a union bound: $\Prb[\Psi] \le |N_x| \cdot e^{-\lambda \cdot \lambda_{\top}} = 16^\ddim \cdot e^{(-c \cdot \ddim)/2} \le 1/2$ for $c \ge 8$, and so
\begin{equation}
	\E[\Lambda\mid\bar{\Psi}]\le\frac{\E[\Lambda]}{\Prb[\overline{\Psi}]}\le\frac{\E[\Lambda]}{1-1/2}=2\cdot(1-e^{-2r\lambda})\cdot e^{2r\lambda}=2\cdot(e^{2r\lambda}-1)~,\label{eq:ExpectedIntersections}
\end{equation}
Since the number of clusters intersecting $B_x$ is $\Lambda+1$, we conclude
\[
\E\left[\#\text{intersections}\right]=1+\E[\Lambda\mid\bar{\Psi}]\le1+2\cdot(e^{2r\lambda}-1)=2\cdot e^{2r\lambda}-1\overset{(*)}{\le}e^{4r\lambda}=e^{O(\ddim\cdot\frac{r}{\Delta})}~.
\]
Here $(*)$ holds as for $a\ge 0$, $e^{2a}-2e^a+1=(e^a-1)^2\ge0$.
\end{proof}

\begin{remark}[Comparison with \cite{Fil24scattering}]
\label{rem:ComparisonFil24}
\Cref{lem:seq} is closely related to the MPX-based construction of
\cite[Theorem~10]{Fil24scattering}. The crucial difference is that
\cite{Fil24scattering} proves a worst-case bound of
$e^{O(\ddim\cdot\frac{r}{\Delta})}\cdot\tilde{O}(\ddim)$ on the number of
intersections, while \Cref{lem:seq} proves an expected bound with the
\EMPH{vanishing-overlap property} (\EMPH{VOP}): as the radius of the ball goes to $0$,
the expected number of intersected clusters goes to $1$.
\cite{Fil24scattering} lacks VOP.
Note that VOP cannot be guaranteed with respect to a worst-case bound:
some infinitesimally small ball must intersect several clusters.
VOP is crucial for our application, as we apply \Cref{lem:seq} repeatedly
over multiple recursion levels. Accordingly, without VOP, the expected
number of intersections would be exponential in the recursion depth.
Moreover, \textbf{\Cref{lem:seq} and its analysis are arguably superior},
as they can be used to prove \cite{Fil24scattering} in a simpler and more
straightforward way, and with better (hidden) constants.
See \appref{apx:comparison-fil24} for additional details.
\end{remark}

\subsection{MPC construction of sparse partition}\label{sec:parallel}
In this section we construct the algorithm (\Cref{alg:par}) creating
consistent hashing in MPC (see \Cref{thm:main_doubling}).
\Cref{alg:par} lifts the sequential subroutine (\Cref{lem:seq}) to the MPC
model via a recursive strategy.
At a high level, we start with a single cluster $\{X\}$.
At each recursion step, we are given a consistent hash $\phiin$ and a set
$\Pin$ of approximate neighborhood hash pairs.
The role of $\Pin$ is to represent all the clusters intersecting the ball
$B_x=B_X(x,\frac{\Delta}{\Gamma})$ around each point $x$.
That is, $\Pin$ induces a cover $\Pi$, which is a collection of possibly
overlapping clusters.
Morally, we iteratively execute the MPX algorithm with ever-decreasing diameter
parameters.
In each iteration, each point $x$, rather than joining only the cluster
corresponding to the winner of the MPX experiment, joins all clusters whose
centers $q$ almost won the competition. Specifically, these are centers whose
$f_x$-gap from the winner is less than $\frac{2\Delta}{\Gamma}$.
Note that this includes all the clusters intersecting $B_x$.
This new overlapping collection of clusters is represented by the new set of
approximate neighborhood hash pairs $\Pout$, on which we continue recursively.

\begin{algorithm}[t]
	\caption{\texttt{MPC-Sparse-Partition}$((X,d_X),\phiin,\Pin,D,\ddim,\Gamma,\Delta,s)$}
	\label{alg:par}
	\DontPrintSemicolon
	\SetKwInOut{Input}{input}\SetKwInOut{Output}{output}
	\Input{Metric space $(X,d_X)$ with doubling dimension $\ddim$ distributed among machines with local space $s$. 
    Consistent hash $\phiin$, and the corresponding approximate neighborhood hash pairs $\Pin$,
    where every cluster induced by $\Pin$ has diameter $\le D$.
    Target scale $\Delta$, and padding parameter $\Gamma\geq 16$.}
	\Output{
		Consistent hash $\phiout$, Approximate neighborhood hash pairs $\Pout$.}
	\BlankLine
	\If{$D\le\Delta$}{ \nllabel{line:base}
		\Return{$(\phiin,\Pin)$} \tcp*[r]{base case}
	}
	$\Pout\leftarrow\emptyset$, 
	$t\gets\lfloor\tfrac{\log s}{3\ddim}\rfloor$;\quad
	$\tilde{\Delta}\gets\max\{8D/2^t,\Delta/4\}$\;\nllabel{line:scale}
	$\Pi\gets\left\{C_q=\left\{ x\mid(x,q)\in\Pin\right\}\mid q\in X \right\}$\tcp*[r]{the collection of clusters induced by $\Pin$}
	\ForEach(\tcp*[f]{in parallel}){cluster $C\in\Pi$ (defined w.r.t. center $q_C$)}{\nllabel{line:clusters}
		construct a $(\tilde{\Delta},\tilde{\Delta}/4)$-net $N_C$ of $C$ via \Cref{lem:net-mpc}\nllabel{line:net}\;
		sample a shift $\delta_q\sim\Texp\bigl(\tfrac{8\ddim}{\tilde{\Delta}},\tfrac{\tilde{\Delta}}{2}\bigr)$ for every $q\in N_C$\nllabel{line:shifts}\;
		broadcast $\{(q,\delta_q)\}_{q\in N_C}$ to all machines holding neighborhood pairs corresponding to $C$\nllabel{line:broadcast}\;
		\For{$x\in C$}{
			let $M_C(x) := \max_{q\in N_C} \delta_q-d_X(x,q)$\; 
			let $q_C^\star(x) := \argmax_{q\in N_C} \delta_q-d_X(x,q)$\nllabel{line:assign}\tcp*[r]{the winner of $x$ w.r.t. $C$}
			\textbf{if} $\phiin(x)=q_C$ \textbf{set} $\phiout(x)\leftarrow q_C^\star(x)$\tcp*[r]{definition of the new consistent hash}
			\For{$q\in N_C$}{
			\textbf{If} $\delta_{q}-d_X(x,q)\ge M_C(x)-\frac{2\Delta}
			{\Gamma}$ \textbf{then} add $(x,q)$ to $\Pout$\nllabel{line:addpairs} \;
			}
		}
        
	}
	\Return{$\texttt{MPC-Sparse-Partition}((X,d_X),\phiout,\Pout,4\tilde{\Delta},\ddim,\Gamma,\Delta,s)$}\;\nllabel{line:recurse}
\end{algorithm}

\paragraph*{Detailed description of \Cref{alg:par}.}
The \texttt{MPC-Sparse-Partition} algorithm 
(\Cref{alg:par}) constructs the hash and the approximate neighborhood hash pairs via a recursive coarse-to-fine clustering of $X$.
The procedure takes as input a consistent hash $\phiin:X\rightarrow X$ and approximate neighborhood hash pairs $\Pin\subseteq X\times X$.
Given $\Pin$, we associate with it the set of clusters $\Pi\gets\left\{C_q=\left\{ x\mid(x,q)\in\Pin\right\}\mid q\in X \right\}$.
That is, the cluster $C_q$ consists of all the points $x$ such that $(x,q)$ is an approximate neighborhood hash pair in $\Pin$. Notice that the clusters $C_q$ might overlap with each other.
In addition, we receive as input a parameter $D$, where it is guaranteed that every cluster in $\Pi$ has diameter at most $D$.
The algorithm is initialized with $\phiin(x)=q'$ for all $x\in X$ and $\Pin=\{(x,q')\}_{x\in X}$,
meaning that all the points are mapped into an arbitrary point $q' \in X$.
Accordingly, $\Pi=\{X\}$.
We also provide a
diameter $D_0$ which is a $2$-approximation of $\diam(X)$ and can be computed in $O(1)$ rounds by a standard algorithm (compute the maximum distance from an arbitrary point and multiply it by $2$).
The base case of the algorithm is when $D\le\Delta$ (\algline{line:base}). Here we terminate and simply return the input $(\phiin,\Pin)$, which already satisfies the diameter guarantee. 

For the recursive case, we process all clusters $C\in\Pi$ in parallel.
For this purpose, we rearrange the pairs in $\Pin$ such that for each cluster $C\in \Pi$, all pairs $(\cdot, q_C)\in \Pin$ are stored in a block of consecutive machines, where $q_C$ is the center of $C$.
We can ensure this layout by sorting them in $O(\log_s n)$ rounds using $O(|\Pin|)$ total space via \Cref{lem:sort}.

Set $t=\lfloor \log s/(3\ddim)\rfloor$ and
$\tilde{\Delta}=\max\{8D/2^t,\Delta/4\}$ (\algline{line:scale}).
For every cluster $C\in\Pi$, we construct a
$(\tilde{\Delta},\tilde{\Delta}/4)$-net $N_C$ of $C$ via \Cref{lem:net-mpc} (\algline{line:net}).
The choice of $\tilde{\Delta}$ makes the net small enough to fit on one machine, while the lower bound
$\tilde{\Delta}\ge\Delta/4$ ensures that the recursion stops with diameter at most $\Delta$.
Indeed, we show below that every cluster created in this step has diameter at most $4\tilde{\Delta}$.
By the packing property and the choice of $\tilde{\Delta}$,
$|N_C|\le (8D/\tilde{\Delta})^{\ddim}\le 2^{t\cdot\ddim}\le s^{1/3}$,
and hence $N_C$ fits on a single machine.
Next, we implement \Cref{lem:seq} in a white-box fashion. 
Specifically, we sample $\delta_q$ for each $q\in N_C$ in that machine (\algline{line:shifts}), and broadcast $\{(q,\delta_q)\}_{q\in N_C}$ to all machines storing pairs from $C$ (\algline{line:broadcast}).
Then, we apply the MPX rule locally: each point $x\in C$ finds the center $q=q_C^*(x)$ maximizing $f_x(q)=\delta_q-d_X(x,q)$ (\algline{line:assign}).
Also, let $M_C(x)=f_x(q)=\delta_q-d_X(x,q)$ with $q=q_C^*(x)$.
As $x$ might belong to more than one cluster of $\Pi$, we set $\phiout(x)=q_C^*(x)$ only when
$\phiin(x)=q_C$.
In this way, every point $x$ of $X$ has its new hash value $\phiout(x)$.

However, we still need to compute a set of approximate neighborhood hash pairs. 
We say that $q'\in N_C$ is an \brick{almost-winner} for $x$ if 
\begin{equation}
\delta_{q'}-d_X(x,q')\ge M_C(x)-\frac{2\Delta}{\Gamma}\label{eq:competitive}
\end{equation}
For every $x\in X$ and $q'\in N_C$ such that
$q'$ is an almost-winner of $x$, 
we add $(x,q')$ to $\Pout$, the newly created approximate neighborhood hash pairs (\algline{line:addpairs}).
Note that the definition of almost-winner is w.r.t. the target scale $\Delta$, and not to the current scale $\tilde{\Delta}$.
This implies that in all but the last level, in the most likely scenario $x$ will join a single pair $(x,\phiout(x))$.
Finally, after processing (in parallel) all the clusters $C\in\Pi$, we make a recursive call with all the collected pairs $\Pout$ and the new hash $\phiout$.

\sloppy
\paragraph*{Diameter bound.}
We prove by induction that
every cluster induced by $\Pin$ has diameter at most $D$
at every recursive call
\texttt{MPC-Sparse-Partition}$((X,d_X),\phiin,\Pin,D,\ddim,\Gamma,\Delta,s)$.
The initial call satisfies this by the choice of $D_0$.

Consider one recursive call, and fix a pair $(x,q)\in\Pout$.
Let $C\in\Pi$ be the cluster in which this pair was added.
Since $N_C$ is a $\tilde{\Delta}$-net of $C$, there is a net point
$q_x\in N_C$ with $d_X(x,q_x)\le\tilde{\Delta}$.
Thus $M_C(x)\ge \delta_{q_x}-d_X(x,q_x)\ge -\tilde{\Delta}$,
as all shifts lie in $[0,\tilde{\Delta}/2]$.
Moreover, $(x,q)$ is added to $\Pout$ only if $q$ is an almost-winner for $x$.
Hence
\[
d_X(x,q)
\le \delta_q-M_C(x)+\frac{2\Delta}{\Gamma}
\le \frac{\tilde{\Delta}}2+\tilde{\Delta}+\frac{2\Delta}{\Gamma}
\le 2\tilde{\Delta},
\]
where the last inequality uses $\Gamma\ge16$ and
$\tilde{\Delta}\ge\Delta/4$.

Now fix a cluster $C_q=\{x\mid (x,q)\in\Pout\}$ induced by $\Pout$.
The above inequality shows that every point of $C_q$ lies within distance
$2\tilde{\Delta}$ from $q$. Therefore $\diam(C_q)\le4\tilde{\Delta}$.
Since the next recursive call uses diameter parameter $4\tilde{\Delta}$,
the induction continues.

\paragraph*{Round and space complexities.}
We first analyze the recursion depth of the algorithm. 
Ignoring the final capped scale, the next diameter parameter is
$D'=4\tilde{\Delta}=\Theta(D/2^t)$, where
$t=\lfloor\tfrac{\log s}{3\ddim}\rfloor$.
Thus the diameter parameter decreases by a factor
$b=\Theta(2^t)$ in each recursive step. 
With local space $s=n^\sigma$ and assuming $\sigma\log n\ge\Omega(\ddim)$,
we have $\log b=\Theta(\frac{\sigma\log n}{\ddim})$.
At the beginning of the algorithm, the diameter parameter is at most
$D_0\le 2\diam(X)$. Since the aspect ratio of $X$ is $n^{O(1)}$,
$\log(D_0/\Delta)=O(\log n)$. Hence the recursion depth is
\[
O\left(\frac{\log(D_0/\Delta)}{\log b}\right)
=
O\left(\frac{\ddim}{\sigma}\right).
\]

Now we analyze the number of rounds used in each recursion step. 
The number of rounds depends on the size of $\Pout$ and $\Pin$. Let $n_\varphi:=|\Pin|+|\Pout|+|X|$. 
We first rearrange the pairs of $\Pin$ to process the clusters of $\Pi$ in parallel. This takes $O(\log_s n_\varphi)$ by \Cref{lem:sort}.
The next step is to construct a net for each cluster. Since we process all clusters in parallel, we can construct a net for every cluster in $O(\log_s^2 n_\varphi)$ rounds in total by \Cref{lem:net-mpc}.
For the broadcasting step, observe that the size of $N_C$ is at most $s^{1/3}$. Thus we can broadcast the shifts for all points of $N_C$ in $O(\log_s n_\varphi)$ rounds in total for all clusters by \Cref{lem:broadcast}.
All remaining steps are local.
Therefore, each recursion step uses $O(\log_s^2 n_\varphi)$ rounds, and the total round complexity
is $O((\ddim/\sigma)\cdot \log_s^2 n_\varphi)$.
The total space complexity is clearly $O(n_\varphi)$.

\begin{theorem}[Consistent Hashing in MPC]\label{thm:main_doubling}
Let $\Delta > 0$, $\sigma\in(0,1]$ and $\Gamma\geq 16$ be fixed, and let $(X,d_X)$ be a metric space with polynomially bounded aspect ratio and doubling dimension~$\ddim$ where $\sigma\log n \ge \Omega(\ddim)$. 
There is an MPC algorithm that, given an $n$-point set $P\subseteq X$ distributed across machines with local memory of size $s = n^\sigma$, and access to a distance oracle for $(X,d_X)$, constructs with high probability both:
\begin{enumerate}
	\item[\textup{(i)}] a $(\Gamma, 2^{O(\ddim/\Gamma)})$-hash $\varphi$ with diameter bound~$\Delta$.
	\item[\textup{(ii)}] a set of approximate neighborhood hash pairs $\tilde{\cP}_\varphi$.
\end{enumerate}
The algorithm uses $O(\ddim/\sigma^3)$ MPC rounds and $\widetilde O(n)\cdot2^{O(\ddim/\Gamma)}$ total space.
\end{theorem}
\begin{proof}

\Cref{alg:par} is a recursive algorithm.
Let $(\brick{\varphi_i},\brick{\Ptilde^i},\brick{D_i})$ be the parameters used for the $i$th recursive call (the other parameters, $(X,d_X),\ddim,\Gamma,\Delta, s$, are fixed).
That is, $\varphi_i$ is the current consistent hash, and $\Ptilde^i$ is the corresponding approximate neighborhood hash pairs. Denote also by \EMPH{$\Pi_i$} the collection of clusters induced by $\Ptilde^i$. $D_i$ is an upper bound on the diameter of every cluster in $\Pi_i$.
Additionally, at initialization $\varphi_0=\Ptilde^0=\{(x,q')\}$, and $D_0$ is an upper bound on $\diam(X)$.
Finally, after $L$ recursive calls $(\varphi_L,\Ptilde^L,D_L)$ is the output of the algorithm (and the result of \Cref{thm:main_doubling}), where $D_L=\Delta$.

Recall that we care here about the intersections of balls with radius $\frac{\Delta}{\Gamma}$ (regardless of the diameter in each scale).
We begin by arguing by induction on $i$ that $\Ptilde^i$ is indeed a set of approximate neighborhood hash pairs for $\varphi_i$, w.r.t. the parameter $\Gamma$.
That is, for every point $z\in B_X(x,\frac{\Delta}{\Gamma})$, it holds that $(x,\varphi_i(z))\in\Ptilde^i$.
The fact that the diameter of every cluster in $\Pi_i$ is bounded by $D_i$ was already argued for.
The base of the induction, $i=0$, holds trivially as there is only a single cluster.
Suppose that the induction holds for $i-1$, and we will prove it for $i$.
To correspond with the notation of \Cref{alg:par}, we denote
$(\phiin,\Pin,D,\Pi)\leftarrow(\varphi_{i-1},\Ptilde^{i-1},D_{i-1},\Pi_{i-1})$ and 
$(\phiout,\Pout,4\tilde{\Delta},\Pi_i)\leftarrow(\varphi_{i},\Ptilde^i,D_i,\Pi_i)$.
Suppose that a cluster centered at $q$ in $\phiout$ intersects the ball $B_X(x,\frac{\Delta}{\Gamma})$ at point $z$.
That is $\phiout(z)=q$.
Let $\phiin(z)=q_z$ be the previous center to which $z$ was mapped, and let $C_z=C_{q_z}\in\Pi$ be the $\phiin$ cluster centered at $q_z$.
Denote also by $N_{C_z}$ the set of net points chosen for $C_z$. In particular $q\in N_{C_z}$.
Denote by $\{\delta_{q'}\}_{q'\in N_{C_z}}$ the sampled shifts of the net points in $C_z$.
By our inductive hypothesis, as $C_z$ intersects $B_X(x,\frac{\Delta}{\Gamma})$ at $z$, it holds that $(x,q_z)\in \Pin$.
Let $M_{C_z}(x) = \max_{q'\in N_{C_z}} \delta_{q'}-d_X(x,q')$.
As $z$ chooses $q$, by \Cref{prop:ball} it follows that $\delta_{q}-d_X(x,q)\ge M_{C_z}(x)-\frac{2\Delta}{\Gamma}$.
In particular, $(x,q)$ was added to $\Pout$, as required.

We conclude that the output of our algorithm is valid. 
It remains to bound the expected number of approximate neighborhood hash pairs (per vertex).
Fix a point $x\in X$ and the ball $B = B_X(x,r)$ with radius $\brick{r} = \Delta/\Gamma$.
Denote by $Z_i=\left|\{q\in X:(x,q)\in\Ptilde^i\}\right|$ the number of pairs containing $x$ in $\Ptilde^i$. The coarsest level has $Z_{0} = 1$ (the entire space as a single cluster), and we want to bound $\E[Z_L]$, which corresponds to the final output of the algorithm.

Denote by $\Delta_i$ the scale used for the MPX algorithms in recursion depth $i$ (i.e., to create $\varphi_i,\Ptilde^i$). That is, $\Delta_i=D_i/4$.
At the last level of recursion, when $i = L$, we used $\Delta_L=\Delta / 4$.
For any scale other than the last scale, we got a significant progress. 
Specifically, the algorithm sets $\Delta_i=\max\{8D_{i-1}/2^t,\Delta/4\}=\max\{32\Delta_{i-1}/2^t,\Delta/4\}$.
Thus for every $i< L$, we have 
$\Delta_i=\frac{32\Delta_{i-1}}{2^t}$. We can assume that $t=\lfloor\tfrac{\log s}{3\ddim}\rfloor\ge 6$ (as $\ddim \le O(\sigma \log n)$ implies $\ddim \le c \cdot \sigma \log n$; by bounding the hidden constant $c \le 1/18$, we strictly guarantee the required geometric reduction $\frac{32}{2^t} \le \frac{1}{2}$). 
We bound $Z_i$ based on $Z_{i-1}$ using conditional expectation.
Fix $Z_{i-1}$.
$x$ belongs to $Z_{i-1}$ different clusters in $\Pi$.
Fix cluster $C\in\Pi_{i-1}$.
The algorithm chooses a net and runs the MPX algorithm using scale $\Delta_i$.
Following the proof of \Cref{lem:seq}, the expected number of centers
$q\in N_C$ satisfying
$\delta_q-d_X(x,q)\ge M_C(x)-\frac{2\Delta}{\Gamma}$
is bounded by $e^{O(\ddim\cdot\frac{r}{\Delta_i})}$.
Indeed, the proof of \Cref{lem:seq} bounds precisely the number of centers whose $f_x$-value is within $2r$ of the maximum, before translating this into a bound on the number of intersected clusters. This is exactly the criterion used
by \Cref{alg:par} to add pairs to $\Ptilde^i$ from the cluster $C$.
We conclude: 
\[
\mathbb{E}\left[Z_{i}\mid Z_{i-1}\right]\leq \sum_{C\in\Pi_{i-1}:C\mbox{ contains x}}e^{O(\ddim\cdot\frac{r}{\Delta_{i}})}=e^{O(\ddim\cdot\frac{r}{\Delta_{i}})}\cdot|Z_{i-1}|~.
\]
Accordingly, by the law of total expectation, 
\[
\mathbb{E}\left[Z_{i}\right]=e^{O(\ddim\cdot\frac{r}{\Delta_{i}})}\cdot\mathbb{E}\left[Z_{i-1}\right]~.
\]
Finally, by induction, 
\[
\mathbb{E}\left[Z_{L}\right]=\prod_{i=1}^{L}e^{O(\ddim\cdot\frac{r}{\Delta_{i}})}\cdot\mathbb{E}\left[Z_{0}\right]=e^{\sum_{i=0}^{L}O(\ddim\cdot\frac{r}{\Delta_{i}})}=e^{O(\ddim\cdot\frac{r}{\Delta})}~,
\]
where in the last inequality we use the fact that $\{\Delta_i\}_{i=0}^{L-1}$ is a geometric progression with jumps $\frac{32}{2^t}\le\frac12$.

It remains to obtain the high-probability bound on the total space claimed in the theorem.
By linearity of expectation, the expected number of stored pairs is at most
$M:=n\cdot 2^{O(\ddim/\Gamma)}$.
To amplify this expectation bound, run
$O(\log n)$ independent trials in parallel, aborting any trial as soon as it stores more than $2M$ pairs.
By Markov's inequality, each trial is aborted with probability at most $1/2$.
Thus, by independence, the probability that all trials are aborted is polynomially small.
Consequently, at least one trial succeeds with high probability.
Since every trial uses $O(M)$ space, the total space is
$M\cdot O(\log n)=\tilde{O}(n)\cdot 2^{O(\ddim/\Gamma)}$.
\end{proof}

\subsection{Neighborhood oracles for consistent hashing}\label{sec:neighborhood-oracle}
Throughout this subsection, $(X,d_X)$ denotes the finite input metric,
with $n=|X|$.
The construction in \Cref{sec:parallel} produces a
$(\Gamma,2^{O(\ddim/\Gamma)})$-hash with diameter bound $\Delta$
and approximate neighborhood hash pairs $\tilde{\cP}_\varphi$.
We now strengthen this representation to support \emph{neighborhood queries}
for every radius $r\in[\Delta/(4\Gamma),\Delta/(2\Gamma)]$.
Here, the underlying metric for neighborhood queries is an \emph{extension} of $(X,d_X)$.

A metric space $(X',d_{X'})$ is called an
\EMPH{extension} of $(X,d_X)$ if $X\subseteq X'$ and
$d_{X'}(x,y)=d_X(x,y)$ for all $x,y\in X$.
For a hash $\varphi:X\to I$, a partition map
$\psi:X'\to I'$ is an \EMPH{extension} of $\varphi$ if
$I\subseteq I'$ and $\psi(x)=\varphi(x)$ for every $x\in X$.
Then we use $\psi|_X$ to denote the restriction of $\psi$ on $X$.
By definition, $\psi|_X=\varphi$.
To define this representation, let $(X',d_{X'})$ be an extension of $(X,d_X)$, and let $\psi:X'\to I'$ be an extension of $\varphi$.
Fix a radius scale $R>0$.

\begin{definition}[Neighborhood oracle]
\label{def:neighborhood-oracle}
A \EMPH{neighborhood oracle} for $\psi$ on $X$ at scale $R$
is a table $\mathcal O\subseteq X\times\psi(X)\times[0,R]$
with at most one record per pair $(x,q)$, such that
for every $x\in X$ and $r\in[R/4,R/2]$,
\[
    \mathcal O(x,r)
    :=\{q:\exists a\le r,\ (x,q,a)\in\mathcal O\}
    =\{q\in\psi(X):B_{X'}(x,r)\cap\psi^{-1}(q)\neq\varnothing\}.
\]
\end{definition}
We call $\mathcal O$ an oracle because it answers the following query:
given $x\in X$ and $r\in[R/4,R/2]$, report the IDs of all clusters
that contain a point of $X$ and intersect $B_{X'}(x,r)$.
The answer is obtained by reporting $q$ for each record
$(x,q,a)\in\mathcal O$ with $a\le r$.

\begin{theorem}\label{thm:main_doubling_with_exact_pairs}
Under the assumptions of \Cref{thm:main_doubling} for a finite input
metric $(X,d_X)$, there is an MPC algorithm that computes $\mathcal O_\Delta$ and $\psi_\Delta|_X$
simultaneously for all dyadic scales $\Delta$ in the distance range of $X$
satisfying the following.
There is a common metric extension $(X',d_{X'})$ of $(X,d_X)$
and randomized maps $\psi_\Delta:X'\to X'$ such that:
\begin{enumerate}
    \item[\textup{(i)}]
    Each $\psi_\Delta$ is a
    $(\Gamma+2,2^{O(\ddim/\Gamma)})$-hash of $(X',d_{X'})$
    with diameter bound $(1+2/\Gamma)\Delta$.
    \item[\textup{(ii)}]
    Each $\mathcal O_\Delta$ is a neighborhood oracle for $\psi_\Delta$
    on $X$ at scale $\Delta/\Gamma$.
\end{enumerate}
The algorithm uses $O(\ddim/\sigma^3)$ MPC rounds and
$\widetilde O(n)\cdot2^{O(\ddim/\Gamma)}$ total space,
as in \Cref{thm:main_doubling}.
\end{theorem}

\paragraph*{MPC Construction.}
For each dyadic scale $\Delta$, put $R=\Delta/\Gamma$ and run
\Cref{alg:par}, executing the runs for different scales in parallel.
Initially, all points belong to a single cluster with center $q_0\in X$,
and we set $\Oin=\{(x,q_0,0):x\in X\}$.

At the beginning of each iteration, we are given a table $\Oin$ along with $\Pin$ and $\phiin$, and we initialize $\Oout=\varnothing$.
Recall that $\Pi$ is a set of clusters obtained by $\Pin$.
For each cluster $C\in\Pi$ with center $q_C$ and each point $x\in C$,
let $b$ denote the radius threshold with $(x,q_C,b)\in\Oin$.
Then for each $q$ in the net $N_C$, let the threshold be
\[
    h_{C,q}(x):=\frac{M_C(x)-(\delta_q-d_X(x,q))}{2}.
\]
Here $\delta_q$ is the shift for $q$ in $C$, and
$M_C(x)=\max_{p\in N_C}\{\delta_p-d_X(x,p)\}$.
Recall that we add $(x,q)$ to $\Pout$ if and only if $h_{C,q}(x)\le R$.
In this case, we add $(x,q,\max\{b,h_{C,q}(x)\})$ to $\Oout$.
Note that the clusters in $\Pi$ may overlap.
Thus, a pair $(x,q)$ might have more than one threshold in the same iteration. In this case, we retain only the record with the smallest threshold.

If this is the last iteration, we set $\varphi_\Delta=\phiout$,
remove from $\Oout$ all records for empty clusters
(those with $q\notin\varphi_\Delta(X)$), and return the remaining table $\mathcal O_\Delta$.
Otherwise, we pass $\Oout$ to the next iteration together with $\phiout$ and $\Pout$.
In the base case, return the current table.

\paragraph*{Analysis.}
The following lemma establishes properties~\textup{(i)} and~\textup{(ii)}.
Each approximate neighborhood hash pair contributes one table record.
Its threshold candidates can be computed locally, and then the minimum is taken in $O(1)$ MPC rounds using Lemma~\ref{lem:sort}.
Thus, the asymptotic round and total-space complexities remain unchanged.

\begin{lemma}\label{lem:neighborhood}
    For the hashes $\varphi_\Delta$ and tables $\mathcal O_\Delta$
    computed by the algorithm, there exists a common metric extension
    $(X',d_{X'})$ and extensions $\psi_\Delta$ of $\varphi_\Delta$ satisfying
    \textup{(i)} and \textup{(ii)} of
    Theorem~\ref{thm:main_doubling_with_exact_pairs}.
\end{lemma}
\begin{proof}
Let $\Delta_1>\cdots>\Delta_m$ be the dyadic scales considered by the
algorithm, and write $R_i=\Delta_i/\Gamma$,
$\varphi_i=\varphi_{\Delta_i}$, and $\mathcal O_i=\mathcal O_{\Delta_i}$.
Recall that each approximate pair $(x,q)$ contributes at most one
record $(x,q,a)$ to $\mathcal O_i$ for each $i$.

For analysis, we define thresholds anew at each iteration, including
for point--hash pairs absent from the stored tables.
Fix a scale $i$ and use the parent clusters, nets, and shifts from the
actual execution of Algorithm~\ref{alg:par}.
Initially, assign threshold $0$ to $(x,q_0)$ for every $x\in X$.
At each iteration, start with an empty set of candidate thresholds.
Consider every nonempty parent cluster $C$ with center $q_C$
and every point $x\in X$, including points outside $C$.
Let $b$ be the threshold of $(x,q_C)$ from the previous iteration.
For each $q\in N_C$, evaluate $h_{C,q}(x)$ by the same formula as above
and generate the candidate threshold $\max\{b,h_{C,q}(x)\}$ for $(x,q)$.
For each pair generated in this iteration, define its new threshold
as the minimum of its candidates from this iteration only, retaining
values greater than $R_i$ as well.
Thus each iteration defines a unique threshold for every $x\in X$
and every cluster center generated in that iteration.
The threshold at termination is denoted by $A_{i,q}(x)$.

At every iteration, the pairs with threshold at most $R_i$ are exactly
those computed by Algorithm~\ref{alg:par}.
This holds initially. Inductively, a new threshold is at most $R_i$
if and only if some parent cluster $C$ with $q\in N_C$ contributes a candidate with
$b\le R_i$ and $h_{C,q}(x)\le R_i$.
By the inductive hypothesis, $b\le R_i$ is equivalent to
$(x,q_C)\in\Pin$, or $x\in C$, so these are precisely the conditions
under which the algorithm adds $(x,q)$ to $\Pout$.
The stored records use the same threshold updates, so the output table is
\[
    \mathcal O_i
    =\bigl\{(x,q,A_{i,q}(x)):
    x\in X,\ q\in\varphi_i(X),\ A_{i,q}(x)\le R_i\bigr\}.
\]

We now define a common metric extension $(X',d_{X'})$.
Start with the complete graph on $X$, with edge lengths given by $d_X$.
For every scale $k\in\{1,2,\ldots,m\}$, points $x,q\in X$, and $t\in[R_k/4,R_k]$,
attach a distinct new vertex $v_{k,x,q,t}$ to $x$ by a single edge
of length $t$; we call $x$ its anchor.
Let $X'$ be the set of all vertices, equipped with the shortest-path metric.
This space may be infinite, but it is used only for analysis.
All added edges have positive length and end at leaves, so distances
within $X$ remain unchanged.

We now extend the hashes $\varphi_i$ to the metric space
$(X',d_{X'})$ constructed above.
For each scale $i$, define $\psi_i:X'\to X'$ by
$\psi_i(x)=\varphi_i(x)$ for $x\in X$ and, for $v=v_{k,x,q,t}$, by
\[
\psi_i(v)=
\begin{cases}
q, & k=i,\ q\in\varphi_i(X),\ A_{i,q}(x)\le t,\\
\varphi_i(x), & \text{otherwise, if }t\le R_i,\\
v, & \text{otherwise}.
\end{cases}
\]
The following claims establish the required properties of this construction.

\begin{claim}\label{claim:diameter}
For every scale $i$, each cluster of $\psi_i$ has diameter at most
$\Delta_i+2R_i$.
\end{claim}
\begin{proof}
For any two points $x, q\in X$ with $q=\psi_i(x)$,
we have $d'(x,q)=d(x,q)\leq \Delta_i/2$ by  \Cref{sec:parallel}.
Now consider a vertex $v=v_{k,x,q,t}$. If it forms a singleton, its diameter is zero. Thus we focus on the first two cases: $\psi_i(v)=q$ and $\psi_i(v)=\varphi_i(x)$.
For the second case, $d'(v,\psi_i(v))=t+d'(x,\psi_i(v))\leq R_i+\Delta_i/2$.
For the first case, $(x,q)$ is an approximate neighborhood pair, and thus  $d'(x,q)\leq \Delta_i/2$ by \Cref{sec:parallel}. Since $d'(x,v)=t\leq R_i$, we have $d'(v,q)\leq R_i+\Delta_i/2$.
Therefore, the diameter of each cluster induced by $\psi_i$ is at most $\Delta_i+2R_i$.
\end{proof}

We say a function $f$ is \EMPH{1-Lipschitz} on $(X,d_X)$ if
$|f(x)-f(y)|\leq d_X(x,y)$ for any two points $x,y\in X$.
The following is a technical claim which will be used to prove property (i).

\begin{claim}
For every scale $i$ and cluster center $q\in\varphi_i(X)$,
the function $A_{i,q}$ is $1$-Lipschitz on $(X,d_X)$.
\end{claim}
\begin{proof}
For each $q\in N_C$, the function $x\mapsto\delta_q-d_X(x,q)$
is $1$-Lipschitz. Taking the pointwise maximum of these functions over $q\in N_C$
preserves this property, so $M_C$ is also $1$-Lipschitz.
Since $h_{C,q}$ is half the difference of these two $1$-Lipschitz
functions, it is $1$-Lipschitz as well.
Starting from the constant zero function, each threshold update takes
pointwise minima and maxima of $1$-Lipschitz functions, preserving
this property.
\end{proof}

\begin{claim}\label{claim:neighborhood}
For each scale $i$, $\mathcal O_i$ is a neighborhood oracle for
$\psi_i$ on $X$ at scale $R_i$.
\end{claim}
\begin{proof}
Fix $x\in X, q\in\psi_i(X)$ and $r\in [R_i/4,R_i/2]$. 
Recall that $A_{i,q}(x)$ is the threshold we defined above.
Now it is sufficient to prove that 
$B_{X'}(x,r)\cap\psi_i^{-1}(q)\neq\varnothing$ if and only if $A_{i,q}(x)\le r$.

Let $z\in B_{X'}(x,r)\cap \psi_i^{-1}(q)$. If $z\in X$, 
we have $A_{i,q}(z)=0$. This is because
the initial threshold is zero, and at each iteration, the actual winner
has $h_{C,q}(z)=0$, so the update preserves a zero threshold.
By the Lipschitz property, 
\[
A_{i,q}(x)\leq d_{X'}(x,z)=d_X(x,z)\leq r.
\]
If $z\notin X$, let $y$ be the anchor of $z$ and
$t:=d_{X'}(z,y)$.
Since $q=\psi_i(z)\in\psi_i(X)\subseteq X$, the third case
in the definition of $\psi_i$ cannot apply to $z$.
In the first case, we have $A_{i,q}(y)\le t$ immediately, and 
in the second case, we have $q=\varphi_i(y)$ and hence $A_{i,q}(y)=0$.
Therefore, by the Lipschitz property,
\[
A_{i,q}(x)\le A_{i,q}(y)+d_X(x,y)
\le t+d_X(x,y)=d_{X'}(x,z).
\]
Thus, $A_{i,q}(x)\le r$ in any case.

Conversely, suppose $A_{i,q}(x)\le r$ and put
$t=\max\{A_{i,q}(x),R_i/4\}$.
Since $r\in[R_i/4,R_i/2]$, we have $t\in[R_i/4,R_i/2]$ and $t\le r$.
In this case, $\psi_i(v)=q$ for $v:=v_{i,x,q,t}$ by the definition of $\psi_i(\cdot)$. Moreover, $v\in B_{X'}(x,r)\cap \psi_i^{-1}(q)$. Therefore,
the converse also holds.
\end{proof}

\begin{claim}
For each scale $i$, $\psi_i$ is a
$(\Gamma+2,2^{O(\ddim/\Gamma)})$-hash of $(X',d_{X'})$.
\end{claim}
\begin{proof}
By \Cref{claim:diameter}, the diameter of each cluster is at most $\Delta_i+2R_i$.
Since 
 $(\Delta_i+2R_i)/(\Gamma+2)=R_i$, it suffices to show that
 the expected number of clusters intersected by 
a nonempty set $S\subseteq X'$ of diameter at most $R_i$ is at most $2^{O(\ddim/\Gamma)}$.
If $S$ contains a leaf of edge length greater than $R_i$, then $S$ is a
singleton, and the bound is immediate.
Otherwise, fix a point $x$ of $S$. If $x\notin X$, we replace it with its anchor.
Every point of $S$ lies in $B_{X'}(x,R_i)$, and
$\psi_i(S)\subseteq\varphi_i(X)$.
For each $q\in\psi_i(S)$, choose $z\in S$ with $\psi_i(z)=q$.
The inequality in the proof of \Cref{claim:neighborhood} gives
$A_{i,q}(x)\le d_{X'}(x,z)\le R_i$.
In the beginning of the proof of \Cref{lem:neighborhood}, we observed that $(x,q)$ belongs to
$\tilde{\cP}_{\varphi_i}$, the set of approximate neighborhood hash pairs
computed by Algorithm~\ref{alg:par} at scale $\Delta_i$.
By the analysis in \Cref{sec:parallel}, the expected number of pairs $(x,\cdot)$ 
in $\tilde{\cP}_{\varphi_i}$ is at most $2^{O(\ddim/\Gamma)}$.
This bounds $\E[|\psi_i(S)|]$ as required.
\end{proof}
These claims establish properties~\textup{(i)} and~\textup{(ii)}
of Theorem~\ref{thm:main_doubling_with_exact_pairs}.
\end{proof}

\section{Applications}
In this section, we provide several applications of \Cref{thm:main_doubling}. 
We observed that recent advances in the MPC model for Euclidean spaces are based solely on consistent hashing (actually on neighborhood hash pairs), and not on other Euclidean properties. Thus they can be extended to other metric spaces, provided that a consistent hashing scheme is available.
Using \Cref{thm:main_doubling}, we observe that for doubling metrics we indeed can provide the required neighborhood hash pairs.

\begin{observation}\label{obs:multiscale}
    Let $\sigma\in(0,1)$ be fixed.
    Then for $\Gamma\in[16,O(\ddim)]$,
    there is an $O(\ddim/\sigma^3)$-round MPC algorithm that, given
    (the labels of) an $n$-point metric space $(X,d_X)$ with doubling dimension $\ddim\leq O(\sigma\log n)$, distributed among machines with local space $s=n^\sigma$, computes approximate neighborhood hash pairs for $(\Gamma, 2^{O(\frac{\ddim}{\Gamma})})$-hashes simultaneously for all distance scales. The total space is $\tilde{O}(n)\cdot 2^{O(\frac{\ddim}{\Gamma})}$, and the
    success probability is at least $1-1/\poly(n)$.
\end{observation}

\subsection{Stochastic tree embeddings}
In this subsection we obtain a stochastic embedding into trees (in fact HSTs) with optimal $O(\log n)$ expected distortion.
A metric embedding is a function $f:X\rightarrow Y$ between the points of two metric spaces $(X,d_X)$ and $(Y,d_Y)$.
A metric embedding $f$ is said to be \EMPH{dominating} if for every pair of points $x,y\in X$, it holds that $d_X(x,y)\le d_Y(f(x),f(y))$. 
The \EMPH{distortion} of a dominating embedding $f$ is $\max_{x \not= y\in X}\frac{d_Y(f(x),f(y))}{d_X(x,y)}$.	
A \EMPH{stochastic embedding} from metric space $(X,d_X)$ into a metric family $\cF$ is a distribution $\cD$ over dominating embeddings from $X$ to $Y\in\cF$. We say that the stochastic embedding $\cD$ has expected distortion $t$ if for every $x,y\in X$, $\E_{(f:X\rightarrow Y)\sim\cD}\left[d_Y(f(x),f(y))\right]\le t\cdot d_X(x,y)$.

An \EMPH{ultrametric} $\left(Z,d\right)$ is a metric space satisfying a
strong form of the triangle inequality, that is, for all $x,y,z\in Z$,
$d(x,z)\le\max\left\{ d(x,y),d(y,z)\right\} $. 
An equivalent definition \cite{BLMN05} is that of an \EMPH{HST} (hierarchically separated tree). HST is a metric space $\left(Z,d\right)$ whose elements
are the leaves of a rooted labeled tree $T$. Each $z\in T$ is associated
with a label $\ell\left(z\right)\ge0$ such that if $q\in T$ is a
descendant of $z$ then $\ell\left(q\right)\le\ell\left(z\right)$
and $\ell\left(q\right)=0$ iff $q$ is a leaf. The distance between
leaves $z,q\in Z$ is defined as $d_{U}(z,q)=\ell\left(\mbox{lca}\left(z,q\right)\right)$
where $\mbox{lca}\left(z,q\right)$ is the least common ancestor of
$z$ and $q$ in $T$.
Every HST embeds into a tree metric isometrically (i.e. with distortion $1$).
As it turns out, HSTs are in fact much more structured than general trees.
Indeed, every HST embeds isometrically into $\ell_2$ (folklore), while general trees require distortion $\Omega(\sqrt{\log\log n})$ \cite{bourgaintrees}.

Every $n$-point metric space stochastically embeds into a distribution over HSTs with expected distortion $O(\log n)$ \cite{FRT04,Bartal04}, which is also tight \cite{Bar96}.
In fact, there is an $\Omega(\log n)$ lower bound already for the grid graph \cite{AKPW95}, which is both planar and has constant doubling dimension.
Stochastic embeddings into trees/HSTs have numerous applications, as they transform an arbitrary metric into a tree (at the cost of expected logarithmic distortion).

Recently, Goranci \etal\cite{goranci2026tree} showed that the algorithm of \cite{FRT04} can be executed on top of a consistent hash (instead of exact distances). In particular, they (implicitly) showed that given approximate neighborhood hash pairs for all distance scales, one can sample a stochastic embedding into HSTs.


\begin{theorem}[Implicit in~\cite{goranci2026tree}]\label{thm:HSTEmbeddingBlackBoxCGJKV24}
Fix $\Gamma\ge2, \Lambda \ge 1$ and
$c\in[1,2]$. 
Let $P$ be a set of $n$ points in a metric space $(X,d_X)$ 
with distances between distinct points in $[1,\Phi]$,
where $\Phi=n^{O(1)}$.
Suppose machines with local memory
$\polylog n\le s\le n$ store $P$ along with 
the hash values on $P$ and a neighborhood oracle $\mathcal O_i$ on $P$ 
at scale $c2^i/\Gamma$ for a $(\Gamma,\Lambda)$-hash of $(X,d_X)$
with diameter bound $c2^i$
 for each $i=0,\ldots,\lceil\log_2\Phi\rceil$.
Then an MPC algorithm can sample a stochastic
embedding of $P$ into HSTs with expected distortion
$O(\Gamma\log\Gamma\cdot\log n)$ in $O(\log_s n)$ rounds,
using total space near-linear in the input size.
\end{theorem}

By plugging \Cref{thm:main_doubling_with_exact_pairs} into \Cref{thm:HSTEmbeddingBlackBoxCGJKV24} above, we conclude that given a metric space with constant doubling dimension distributed among machines with local space $n^\sigma$ for a constant $\sigma$, in $O(1)$ rounds we can sample an embedding into an HST with optimal distortion $O(\log n)$, while using near-linear total space.
More generally, our algorithm works for any metric space with doubling dimension $\ddim\leq O(\sigma\log n)$ as stated below.

\begin{corollary}\label{cor:HST}
    Let $\sigma\in(0,1)$ be fixed.
    Then for every $\Gamma\ge 16$, there is an MPC algorithm that, given (the labels of) an $n$-point metric space $(X,d_X)$ with doubling dimension $\ddim\leq O(\sigma \log n)$ distributed among machines with local space $s=n^\sigma$, 
    samples an embedding into an HST with expected distortion $O(\Gamma\log \Gamma\cdot \log n)$ 
    in $O(\ddim/\sigma^3)$ rounds. The total space used is $\tilde{O}(n)\cdot 2^{O(\frac{\ddim}{\Gamma})}$, and
    the success probability is at least $1-1/\poly(n)$. 
\end{corollary}

\paragraph*{Applications via stochastic tree embeddings.}
A stochastic tree embedding yields $O(\log n)$-approximation MPC algorithms for many optimization problems, including Earth Mover's Distance (EMD). 
Here, the standard approach is to embed the metric into an HST, and solve the problem on the HST.
In particular, 
computing EMD on the HST yields an $O(\log n)$-approximation in the original metric, and can be performed in $O(1)$ rounds in the fully scalable MPC model~\cite{goranci2026tree}.
For completeness, we include details in~\appref{apd:EMD}.

\begin{corollary}\label{cor:emd}
    Let $\sigma\in(0,1)$ be fixed.
    For every $\Gamma\geq 16$, there is an MPC algorithm such that, given two $n$-point sets $A,B$ in a point metric space with doubling dimension $\ddim\leq O(\sigma\log n)$ distributed across machines with local space $s=n^\sigma$, computes in $O(\ddim/\sigma^3)$ rounds an $O(\Gamma\log\Gamma \cdot \log n)$-approximation to the Earth Mover's Distance between $A$ and $B$.
    The total space used is $\tilde{O}(n)\cdot 2^{O(\frac{\ddim}{\Gamma})}$, and the success probability is at least $1-1/\poly(n)$.
\end{corollary}

Although this approach follows a standard reduction, it already yields the first MPC algorithm for EMD in doubling metrics. 
In contrast, there are several previous works on EMD in Euclidean spaces. 
In a constant-dimensional Euclidean space, a $(1+\eps)$-approximation to the EMD cost can be computed in $O(1)$ rounds in the fully-scalable MPC model~\cite{andoni2014}.
In high-dimensional Euclidean spaces, the best-known result for estimating the EMD cost in $O(1)$ rounds relies on tree embeddings and incurs an $O(\log n)$-approximation factor~\cite{goranci2026tree}.

\subsection{Facility location}
In this subsection we obtain a constant-factor approximation for the facility location problem (with uniform facility cost) in a constant number of rounds in the sublinear MPC model.

Given a metric space $(X,d_X)$, a uniform opening cost $\wopen > 0$ and some constant $z \geq 1$, 
the objective function of  \EMPH{Power-$z$ (Uniform) Facility Location} for a set of facilities $F \subset X$ is defined as 
\[
    \fl_z(P, F) \eqdef |F| \cdot \wopen + \sum_{x \in P} d^z_X(x,F)~,
\]
where $d_X(x,F)=\min_{y\in F}d_X(x,y)$. We denote the minimum value of a solution for the power-$z$ (uniform) facility location by  
 $\OPTfl_z(P) \eqdef \min_{F\subseteq X}\fl_z(X,F)$.
Czumaj \etal \cite{CGJKV24} implicitly proved the following meta theorem.

\begin{theorem}[Implicit in \cite{CGJKV24}]\label{thm:facilityBlackBoxCGJKV24}
    Consider an $n$-point metric space $(X,d_X)$ where all the distances are between $1$ and $\Phi$, and suppose that there are machines with local memory $s\ge\polylog(n)$ that for every $i\in[\log\Phi]$ store approximate neighborhood hash pairs $\tilde{\cP}_{\varphi_i}$ for $(\Gamma,\Lambda)$-hash with diameter bound $2^i$. 
    Then there is an MPC algorithm in $O(\log_{s}n)$ rounds that with probability at least $1-1/\poly(n)$ computes a $\Gamma^{O(z)}$-approximate solution to the power-$z$ (uniform) facility location problem. The total space used is $\tilde{O}(\sum_{i=0}^{\log\Phi}|\tilde{\cP}_{\varphi_i}|)$.
\end{theorem}

By plugging \Cref{obs:multiscale} into \Cref{thm:facilityBlackBoxCGJKV24} above, we conclude that given a metric space with constant doubling dimension distributed among machines with local space $n^\sigma$ for a constant $\sigma$, in $O(1)$ rounds we can compute an $O(1)$-approximation to facility location, while using near-linear total space.

\begin{corollary}\label{cor:facilityLocation}
    Let $\sigma\in(0,1)$ be fixed. 
    Then for every $\Gamma\ge 16$, there is an MPC algorithm that,
    given (the labels of) an $n$-point metric space $(X,d_X)$ with doubling dimension $\ddim\leq O(\sigma\log n)$ distributed among machines with local space $s=n^\sigma$, 
    computes a $\Gamma^{O(z)}$-approximate solution to the power-$z$ (uniform) facility location problem
    in $O(\ddim/\sigma^3)$ rounds. The total space used is $\tilde{O}(n)\cdot 2^{O(\frac{\ddim}{\Gamma})}$, and the 
    success probability is at least $1-1/\poly(n)$.
\end{corollary}

\subsection{\texorpdfstring{$k$-median and $k$-means}{k-median and k-means}}
In this subsection we obtain a bicriteria approximation for $k$-median and $k$-means in a constant number of rounds in the sublinear MPC model.
In fact, the approximation is for a more general $(k,z)$-objective.
In the context of the sublinear MPC model, the case where $k$ is smaller than local memory $s$ can be solved using core-set techniques (construct core-sets locally in every machine, and gradually compose them into a single machine).
The case of large $k$ is much more challenging. Thus here, we consider the case where $k$ is larger than the local memory space.

Given a metric space $(X,d_X)$, an integer $k \geq 1$ and some $z \geq 1$,
the objective function of \EMPH{$(k,z)$-Clustering}
for a center set $C \subset X$ with $|C| \leq k$ is defined as
\[
  \cL_z(X, C) \eqdef \sum_{x \in X} d^z(x, C).
\]
Notice that the special cases $z = 1$ and $z = 2$ are called \EMPH{$k$-Median} and \EMPH{$k$-Means}, respectively.
We denote the minimum value of a solution for $(k,z)$-clustering by
$\OPTkz(X) := \min_{C\subseteq X: |C|\le k}\cL_z(X, C)$. 
Czumaj \etal \cite{CGJKV24} implicitly proved the following meta theorem.

\begin{theorem}[Implicit in \cite{CGJKV24}]\label{thm:kmedianBlackBoxCGJKV24}
    Consider a metric space $(X,d_X)$ where all the distances are between $1$ and $\Phi$, and suppose that there are machines with local memory $s\ge\polylog(n)$ that for every $i\in[\log\Phi]$ 
    store approximate neighborhood hash pairs $\tilde{\cP}_{\varphi_i}$ for $(\Gamma,\Lambda)$-hash with diameter bound $2^i$. \hfill
    Then for every $k\in\N,z\ge1,\eps\in(0,1)$, there is an MPC algorithm in $O(\log_{s}n)$ rounds that with probability at least $1-1/\poly(n)$ computes a set $C\subseteq X$ of $|C|\le(1+\eps)k$ centers whose connection cost is 
    $\cL_z(X, C) \le \Gamma^{O(z)}\eps^{-2}\cdot \OPTkz(X)$,
    where $\OPTkz(X)$ is the optimal $(k,z)$-clustering cost.
\end{theorem}

By plugging \Cref{obs:multiscale} into \Cref{thm:kmedianBlackBoxCGJKV24} above, we conclude that given a metric space with constant doubling dimension distributed among machines with local space $n^\sigma$ for a constant $\sigma$, in $O(1)$ rounds we can compute a $(1+\eps, O(\eps^{-2}))$-bicriteria approximation to the $k$-median and $k$-means, while using near-linear total space.

\begin{corollary}\label{cor:kmedian}
    Let $\sigma\in(0,1)$ be fixed.
    Then for every
    $k\in\N,z\ge1,\Gamma\ge 16,\eps\in(0,1)$, there is an MPC algorithm that,
    given (the labels of) an $n$-point metric space $(X,d_X)$ with doubling dimension $\ddim\leq O(\sigma\log n)$ distributed among machines with local space $s=n^\sigma$, 
    computes a set $C\subseteq X$ of $|C|\le(1+\eps)k$ centers whose connection cost is 
    $\cL_z(X, C) \le \Gamma^{O(z)}\eps^{-2}\cdot \OPTkz(X)$ in $O(\ddim/\sigma^3)$ rounds,  
    where $\OPTkz(X)$ is the optimal $(k,z)$-clustering cost.
    The total space used is $\tilde{O}(n)\cdot 2^{O(\frac{\ddim}{\Gamma})}$, and 
    the success probability is at least $1-1/\poly(n)$.
\end{corollary}

\subsection{Minimum-size clustering}
In this subsection, we introduce a fully scalable $O(\log\log\log n)$-approximate algorithm for the $r$-gather problem, also known as minimum-size clustering, in $O((\log\log n)^2)$ rounds,
as well as more general trade-offs between the approximation ratio and the round complexity.
The $r$-gather problem can be viewed as the $k$-center problem with a minimum cluster-size constraint.  
More formally, this problem is defined as follows.
Given a metric space $(X, d_X)$ and an integer $r \ge 1$, the \EMPH{$r$-gather} problem asks for a partition of $X$ into clusters, each of size at least~$r$, minimizing the maximum cluster radius. More precisely, for a partition $P_1, \ldots, P_t$ of $X$ with centers $c_1, \ldots, c_t$, the cost is $\max_{i} \max_{p \in P_i} d_X(p, c_i)$. We denote the optimal $r$-gather cost by $\hat{\rho}$. 
A minimum cluster-size constraint naturally arises in clustering applications for privacy and anonymization, where extremely small clusters may reveal sensitive information about individuals~\cite{aggarwal2010achieving,byun2007efficient}. 
From this motivation, this problem is introduced independently by \cite{karget2000building} and \cite{guha2000hierarchical}.
This problem is NP-hard in general metric spaces.
Armon~\cite{armon2011min} gave 
a tight $3$-approximation in a general metric space.
Recently, this problem has also been studied in several settings including the MPC and dynamic setting~\cite{epasto2022massively}. More specifically, 
\cite{epasto2022massively} gave a fully scalable MPC algorithm computing an $O(1)$-approximate $r$-gather solution for points in $\mathbb{R}^d$ in $O(\log^\varepsilon n)$ rounds and $n^{1+\gamma+o(1)}$ total space. Their construction is confined to Euclidean space, as it builds the near-neighbor graph via Euclidean locality sensitive hashing.

Our strategy here is to lift this algorithm to general doubling metrics 
using the neighborhood hash pairs of \Cref{obs:multiscale}.
The MPC algorithm of \cite{epasto2022massively}
consists of two steps: constructing an $C$-approximate $(R,r)$-near neighbor graph $G$ and computing a $\beta$-ruling set\footnote{
For a graph $G = (V, E)$ and an integer parameter $\beta \ge 2$ be an integer, a subset $S \subseteq V$ is a \EMPH{$\beta$-ruling set} of $G$ if: (1) for every $u, v \in S$ with $u \ne v$, $\mathrm{dist}_G(u, v) \ge 2$, and (2) for every $v \in V$, $\mathrm{dist}_G(v, S) \le \beta$.} on the square graph of $G^2$. 
For an integer $k$, we let $G^{k}$ denote the \EMPH{$k$-th power graph} of $G$ that has the same set of vertices as $G$, but 
in which two vertices are connected when their hop-distance in $G$ is at most $k$.
In particular, $G^2$ is called the \EMPH{square graph} of $G$.

\begin{definition}[Near-neighbor graph {\cite[Definition 3.2]{epasto2022massively}}]\label{def:near-neighbor}
 Let $(X, d_X)$ be a metric space and $R, r \ge 1$, $C \ge 1$. A graph $G' = (X, E')$ is a \EMPH{$C$-approximate $(R, r)$-near neighbor graph} if:
  \begin{enumerate}
    \item For every edge $(u, v) \in E'$, $d_X(u, v) \le C \cdot R$.
    \item For every vertex $v \in X$, either $|N_{G'}[v]| \ge r$, or $B_X(v, R) \subseteq N_{G'}[v]$,
  \end{enumerate}
  where $N_G[v]=\{u\in V(G)\mid (v,u)\in E(G)\}\cup \{v\}$ is the closed set of neighbors of $v$.
\end{definition}

More specifically, this algorithm first computes a $C$-approximate $(R,r)$-near-neighbor graph $G$ with $O(\log n)$ different scales $R$. Here, $r$ is given as an input,
while $C$
controls the approximation factor of the final $r$-gather solution.
As once these near neighbor graphs $G$ are constructed 
from the metric space,
the subsequent step of computing a ruling set on $G^2$ 
does not require further access to the metric space.
Therefore, the original problem reduces to computing an approximate near-neighbor graph $G$.

In our setting, however, the size of the approximate near-neighbor graph $G(R)$ we compute for each scale $R$ has quadratic size. Thus we cannot store it explicitly. 
Instead, we construct a compact representation of $G(R)$:
we have a sparse graph $H(R)$ and its square graph $(H(R))^2$ is a near-neighbor graph $G(R)$. 
This introduces one additional issue: 
since the original algorithm computes a ruling set on $G(R)^2$,
we need to compute a ruling set for the 4th power graph of $H(R)$. Fortunately, the algorithm of \cite{epasto2022massively} for computing a ruling set is sufficiently general for our purpose: 
It can be used for computing a ruling set of the $k$-th power graph of an input graph for any fixed $k$.
The round complexity depends on the maximum degree $\Delta_k$  of $G^k$,
but even the trivial bound $\Delta_k\leq n$ gives a space complexity sufficient for our purpose. 
This reduction is summarized in the following lemma.

\begin{lemma}[Implicit in \cite{epasto2022massively}]\label{lemma:epasto2022massively}
    Let $\sigma \in (0, 1)$ be a fixed constant. Then for every $r \ge 1$, every $\Gamma \in [4, \ddim]$, and arbitrarily small parameters $\varepsilon, \gamma > 0$.
    Consider an $n$-point metric space $(X,d_X)$ of doubling dimension $\ddim$ and all distances are between $1$ and  $\Phi=n^{O(1)}$.
    For each scale parameter $R\in\{1,2,4,\ldots,\Phi\}$,
    a graph $H(R)$ such that $(H(R))^2$ is an $O(\Gamma)$-approximate $(R,r)$-near neighbor graph of $X$ is initially distributed across different machines.
    Then there is an MPC algorithm that 
    computes an $O(\Gamma \log(1/\varepsilon))$-approximate $r$-gather solution in $O\!\left(\frac{\log(1/\varepsilon)}{\gamma}\log^\varepsilon n \log\log n\right)$ rounds, using $O(n^{\sigma})$ local space and 
    $\tilde O\big((|E_H|+n)\,n^\gamma\big)$
    total space. The success probability is at least $1 - 1/\mathrm{poly}(n)$.     
\end{lemma}

Now in the following lemma, we show how to compute a sparse representation of an $O(\Gamma)$-approximate $(R,r)$-near neighbor graph using the neighborhood hash pairs  obtained from \Cref{obs:multiscale}.
\begin{lemma}\label{lem:nng}
      Let $\sigma \in (0, 1)$, $r \ge 1$, $\Gamma \in [4, \ddim], R\in\{1,2,4,\ldots,\Phi\}$ and $\varepsilon, \gamma > 0$.
      Let $(X,d_X)$ be an $n$-point metric space of doubling dimension $\ddim$.      
      There is an MPC algorithm that computes a graph $H(R)$
      of size $O(n)\cdot 2^{O(\mathrm{ddim}/\Gamma)}$
      such that $(H(R))^2$ is     
      an $O(\Gamma)$-approximate $(R,r)$-near neighbor graph
      of $X$ 
      in $O\!\left(\frac{\ddim}{\sigma^3}\right)$ rounds, using $O(n^{\sigma})$ local space and $n\cdot 2^{O(\ddim /\Gamma)}$ total space. The success probability is at least $1 - 1/\mathrm{poly}(n)$.
\end{lemma}
\begin{proof}
Let $\varphi$ be a $(\Gamma, 2^{O\left(\frac{\ddim}{\Gamma}\right)})$-hash, and 
$\tilde{\mathcal{P}}_\varphi$ be the set of approximate neighborhood hash pairs obtained from
 \Cref{obs:multiscale} with diameter bound $\Delta = 2\Gamma R$. 
Recall that $\tilde{\mathcal{P}}_\varphi$ satisfies: (i) $\tilde{\mathcal{P}}_\varphi \supseteq \mathcal{P}_\varphi$, where $\mathcal{P}_\varphi = \{(v,q) : \exists\, y \in B_X(v, \Delta/\Gamma) \text{ with } \varphi(y) = q\}$; and (ii) for every $(v,q) \in \tilde{\mathcal{P}}_\varphi$, $\varphi^{-1}(q) \subseteq B_X(v, 2\Delta)$.
For each hash value $q$ in the image of~$\varphi$, fix an arbitrary representative $\mathrm{rep}(q) \in \varphi^{-1}(q)$. Define the graph $H=(X,E_H)$ by
\[
  E_H \;=\; \bigl\{(v, \mathrm{rep}(q)) : (v, q) \in
  \tilde{\mathcal{P}}_\varphi\bigr\}.
\]
The size of $E_H$ is $O(n)\cdot 2^{O(\mathrm{ddim}/\Gamma)}$.
This is because each pair $(v,q) \in \tilde{\mathcal{P}}_\varphi$ contributes at most one edge to $E_H$, and $|\tilde{\mathcal{P}}_\varphi| = O(n)\cdot 2^{O(\mathrm{ddim}/\Gamma)}$.

In the following, it suffices to show that $H^2$ is an $O(\Gamma)$-approximate $(R,r)$-near neighbor graph of $X$. 
First, we show that for every edge $(u,v) \in E(H^2)$,\;
          $\dist_{X}(u, v) \le 8 \Gamma R$.
    For every $(v, \mathrm{rep}(q)) \in E_H$, we have $(v, q) \in \tilde{\mathcal{P}}_\varphi$, so $\mathrm{rep}(q) \in \varphi^{-1}(q) \subseteq B_X(v, 2\Delta) = B_X(v, 4\Gamma R)$. Thus $\dist_{X}(v, \mathrm{rep}(q)) \le 4\Gamma R$. For an edge $(u,v)$ in $H^2$, there exists a vertex $w$ with $(u,w), (w,v) \in E_H$, so by the triangle inequality, $\dist_{X}(u,v) \le 8\Gamma R$.

Now we show that for every vertex $v\in X$,\;
          $B_{X}(v, R) \subseteq N_{H^2}[v]$.
Notice that in Definition~\ref{def:near-neighbor},
it is sufficient to satisfy one of the two conditions for every vertex $v\in X$. But we show that for every $x\in X$,
the second condition always holds in our case. 
    Let $v \in X$ and $u \in B_{X}(v,R)$. Since $\Delta/\Gamma = 2R > R$, we have $u \in B_X(v, \Delta/\Gamma)$, and therefore $(v, \varphi(u)) \in \mathcal{P}_\varphi \subseteq \tilde{\mathcal{P}}_\varphi$. Hence, $(v, \mathrm{rep}(\varphi(u))) \in E_H$. Similarly, $u \in B_X(u,0) \subseteq B_X(u, \Delta/\Gamma)$, so $(u, \varphi(u)) \in \mathcal{P}_\varphi \subseteq \tilde{\mathcal{P}}_\varphi$, giving $(u, \mathrm{rep}(\varphi(u))) \in E_H$. Both $v$ and $u$ are adjacent to $\mathrm{rep}(\varphi(u))$ in $H$, so $\mathrm{dist}_H(u, v) \le 2$, that is, $(u, v) \in E(H^2)$.
    Therefore, $H^2$ is an $O(\Gamma)$-approximate $(R,r)$-near neighbor graph of $X$. 

    Now we analyze the round and space complexities of the algorithm. 
    Using \Cref{obs:multiscale}, we compute $\tilde{\mathcal{P}}_\varphi$ in $O(\ddim/\sigma^3)$ rounds with $O(n^{\sigma})$ local space and $n\cdot 2^{O(\ddim/\Gamma)}$ total space. Selecting representatives requires one sorting step on the hash values, taking $O(1/\sigma)$ rounds. Forming the edges of $H$ requires, for each pair $(v,q) \in \tilde{\mathcal{P}}_\varphi$, looking up $\mathrm{rep}(q)$, which is another sorting step in $O(1/\sigma)$ rounds.
\end{proof}

By combining Lemma~\ref{lemma:epasto2022massively} with Lemma~\ref{lem:nng}, we have the following corollary.
\begin{corollary}\label{cor:r-gather}
    Let $\sigma \in (0, 1)$ be a fixed constant. Then for every $r \ge 1$, every $\Gamma \in [4, \ddim]$, and arbitrarily small parameters $\varepsilon, \gamma > 0$, given an $n$-point metric space $(X,d_X)$ of doubling dimension $\ddim$, there is an MPC algorithm that computes an $O(\Gamma \log(1/\varepsilon))$-approximate $r$-gather solution in $O\!\left(\frac{\ddim}{\sigma^3} + \frac{\log(1/\varepsilon)}{\gamma}\log^\varepsilon n \log\log n\right)$ rounds, using $O(n^{\sigma})$ local space and $n^{1+\gamma+o(1)}\cdot 2^{O(\ddim /\Gamma)}$ total space. The success probability is at least $1 - 1/\mathrm{poly}(n)$.
\end{corollary}

By setting $\Gamma=O(1)$ and $\varepsilon=1/\log\log n$,
we have the following corollary. 

\begin{corollary}\label{cor:r-gather-nice}
    Let $\sigma \in (0, 1)$ be a fixed constant. Then for every $r \ge 1$, and for every small constant $\gamma < 1$, given an $n$-point metric space $(X,d_X)$ of a constant doubling dimension, there is an MPC algorithm that computes an $O( \log\log\log n)$-approximate $r$-gather solution in $O\!\left( (\log\log n)^2\right)$ rounds, using $O(n^{\sigma})$ local space and $n^{1+\gamma}$ total space. The success probability is at least $1 - 1/\mathrm{poly}(n)$.
\end{corollary}

\section{Lower bounds}
In this section, we provide lower bounds on the number of rounds for stochastic tree embeddings and consistent hashing. Recall that we work in the distance-labeled MPC model, where the input is an $n$-point subset $P \subseteq X$ distributed across machines, and the algorithm can access the metric only via distance queries of the form $\dist(u,v)$. In particular, we assume that no additional structural information (e.g., labels, embeddings, or sketches of the points) is available, and hence all information about the metric must be obtained through such queries. 

\subsection{Stochastic tree embeddings}
As discussed above, stochastic tree embeddings are a central primitive in geometric MPC algorithms as they enable reducing many optimization problems on general metrics to tree metrics.
Consequently, a fully scalable MPC algorithm for stochastic embeddings with small distortion yields efficient MPC algorithms for a wide range of optimization problems. 
While such embeddings are efficiently computable in Euclidean spaces, it is natural to ask whether similar algorithms can be achieved for general metrics. 

In this subsection, we show that no such algorithm exists for general metrics in the fully scalable MPC model.
In particular, we show that any fully scalable MPC algorithm for sampling an embedding of a general metric into a tree with polynomially bounded distortion using $O(n^{1+\eps})$ total space and $O(n^{\sigma})$ local space requires a polynomial number of rounds.
More generally, we show that 
this lower bound can be extended to any metric space with doubling dimension $\ddim\geq \Omega(\sigma\log n)$.

\medskip 
Our lower bounds for stochastic tree embedding and consistent hashing come from the query complexity for the closest pair problem: Given a metric space $(X,\dist)$, the goal is to find two points $p,q\in X$ that minimize $\dist(p,q)$.
Indyk~\cite{indyk1999sublinear} showed that any $c(n)$-approximation algorithm for the closest pair problem requires $\Omega(n^2)$ queries for a metric space with doubling dimension $\Theta(\log n)$, where $c(n)$ is a polynomial function of $n$.
We extend this lower bound to any metric space with doubling dimension $\ddim$. 
\begin{lemma}\label{lem:lb-closest}
    Let $c(n)$ be a polynomial function of $n$
    with $c(\cdot)\geq 1$.
    Any randomized algorithm with access to a distance oracle of an $n$-point metric space of doubling dimension $\ddim$, where all distances are polynomial, that computes a $c(n)$-approximate closest pair with probability at least $2/3$ 
    requires $\Omega(n\cdot 2^ \ddim)$ queries in the worst case. 
\end{lemma}
\begin{proof}
We use Yao's minimax principle. For this purpose, we first define a distribution $\mathcal D$ over distance functions on a set $P$ of $n$ points. We simply let $c:=c(n)$.
We partition $P$ into subsets $P_1,P_2,\ldots,P_t$ of size $2^\ddim$ with $t=n/2^\ddim$.
A sample from $\mathcal D$ is obtained by choosing an index $i^*$ uniformly at random and then by choosing an unordered pair $e$ of points of $P_{i^*}$ uniformly at random.
Then the distance function $\dist_{e}$ is defined as follows. 

\[\dist_{e}(x,y)=
\begin{cases}
0  & x=y \\
1 & \text{if } (x,y)=e \\
10\cdot c  & \text{if } x,y\in P_i, \\
100\cdot c\cdot |i-j| & \text{if } x\in P_i \text{ and } y\in P_j \text{ with }i\neq j
\end{cases}
\]
Note that $(P,\dist_{e})$ is a metric.
Clearly, a $c$-approximate closest pair is $e$, and it is a unique solution.
Moreover, the doubling dimension of $(P,\dist_e)$ is at least $\ddim-1$ and at most $\ddim+1$. 
To see this, observe that any ball of radius  $10c$ contains exactly $2^\ddim$ points.
These points are contained in the same subset $P_i$. Thus to cover this ball, we need at least $2^\ddim-1$ balls of half radius.
Therefore, the doubling dimension is at least 
$\ddim-1$.
For an upper bound,
observe that a ball of radius less than $10c$
contains at most two points, thus it can be covered by two balls of half radius.
A ball of radius $R$ with $10c\leq R <100c$ can be covered by $2^\ddim$ balls of half radius.
A ball $B$ of radius at least $100c$ can be covered by at most three balls of half radius. More specifically, we choose three sets among $P_1,\ldots,P_t$: the set contained in $B$ with smallest index, the set containing the center of $B$, and the set contained in $B$ with largest index. Then we pick an arbitrary point from each of the three sets. The balls centered at these points with half radius contain $B$.
Therefore, in any case, a ball can  be covered by $\max\{2^\ddim,3\}$ balls of half radius,
and thus the doubling dimension is at most $\ddim+1$.

Consider a deterministic algorithm $\mathcal A$ that detects $e$ with probability at least $2/3$. 
We prove a stronger claim: this holds even if $\mathcal A$ knows the partition $(P_1,\ldots,P_t)$ in advance. 
Assume to the contrary that $\mathcal A$ uses  
$q=o(n\cdot 2^\ddim)$ oracle queries in the worst case. Since $\mathcal A$ knows the partition in advance, we may assume that for a query of $(x,y)$ made by $\mathcal A$,
$x,y$ are contained in the same subset $P_i$. 
We also slightly modify $\mathcal A$ so that if it queries $e$, then it returns $e$ immediately and terminates. 
Let $\mathcal S$ be the longest sequence of queries made by $\mathcal A$. 
Since $\mathcal A$ is deterministic, and it terminates immediately after it reveals $e$,
the sequence of queries made by $\mathcal A$ is a prefix of $\mathcal S$ for any choice of $e$. Note that the number of pairs contained in the same subset is $\Theta(n\cdot 2^\ddim)$.
Since $q=o(n\cdot 2^\ddim)$, there always exist $(1-o(1))n\cdot 2^\ddim$ pairs $e_1,\ldots,e_\ell$ of points with $e_1,\ldots,e_\ell\notin \mathcal S$ such that the two points in each of the pairs are contained in the same subset. Then $\mathcal A$ returns the same output for inputs $\dist_{e_i}$ for all $i\in\{1,\ldots,\ell\}$. Thus the failure probability of $\mathcal A$ is larger than $\frac{(1-o(1))n\cdot 2^\ddim-1}{n\cdot 2^\ddim} > 1/3$. Therefore, this implies that any deterministic algorithm that detects $e$ with probability at least $2/3$ uses $\Omega(n\cdot 2^\ddim)$ queries.

Therefore, by Yao's minimax principle, any randomized algorithm that computes a $c$-approximate closest pair with probability at least $2/3$ requires $\Omega(n\cdot 2^\ddim)$ queries in the worst case.
Notice that this lower bound holds even in the case that an input always comes from the support of $\mathcal D$.
\end{proof}

Now we are ready to prove a lower bound for stochastic embeddings. In a general metric space, a fully scalable MPC algorithm for this problem requires a polynomial number of rounds. 

\begin{theorem}\label{thm:LB-hst}
    Let $c(n)$ be any polynomial function of $n$, and $\sigma,\eps\in(0,1)$ be fixed.
    Assume that there is an MPC algorithm $\mathcal A$ with access to a distance oracle of an $n$-point metric space of doubling dimension $\ddim$ that  samples a stochastic embedding into a HST with expected distortion $c(n)$ using $O(n^{1+\eps})$ total space and $O(n^\sigma)$ local space. 
    Then the number of rounds used by $\mathcal A$ is $\Omega(2^\ddim/n^{\eps+\sigma})$.
    In particular, for a general metric space where $\ddim=\log n$, the number of rounds used by $\mathcal A$ is $\Omega(n^{1-\sigma-\eps})$.
\end{theorem}
\begin{proof}
    Let $(X,\dist)$ be an $n$-point metric space with doubling dimension $\ddim$. Let $T$ be a HST returned by $\mathcal A$ on $(X,\dist)$. 
    Let $R$ denote the number of rounds used by $\mathcal A$ to compute $T$. 
    Then the closest pair in $(X,\dist_T)$ can be found without making any distance queries for $(X,\dist)$ as follows.
    Let $(p,q)$ be the closest pair in $(X,\dist_T)$,
    and let $(p^*,q^*)$ be the closest pair in $(X,\dist)$.
    We have $\dist(p,q)\leq O(c(n))\cdot \dist(p^*,q^*)$ with probability at least $2/3$.
    This is because, with probability at least $2/3$,
    we have 
    $\dist_T(p^*,q^*)\leq O(c(n))\cdot \dist(p^*,q^*)$.
    Conditioned on this event, we have 
    \[
    \dist(p,q) \leq \dist_T(p,q) \leq \dist_T(p^*,q^*) \leq O(c(n))\cdot \dist(p^*,q^*),
    \]
    and thus $(p,q)$ is an $O(c(n))$-approximate closest pair of $(X,\dist)$. 
    
    Here, $\mathcal A$ uses $O(R\cdot n^{1+\varepsilon+\sigma})$ distance queries. 
    To see this, observe that, in each round, each machine makes at most $O(n^{2\sigma})$ distance queries as it stores at most $n^\sigma$ points. Since the number of machines is $O(n^{1+\varepsilon-\sigma})$ and the algorithm performs $O(R)$ rounds,
    the distance queries made by $\mathcal A$ are $O(R\cdot n^{1+\varepsilon+\sigma})$. 
    In other words, we can compute an $O(c(n))$-approximate closest pair by issuing $O(R\cdot n^{1+\varepsilon+\sigma})$ queries.
    By~\Cref{lem:lb-closest}, this quantity should be
    $\Omega(n\cdot 2^\ddim)$.
    This implies that
    $R\geq \Omega(2^\ddim / n^{\varepsilon+\sigma})$
\end{proof}

\paragraph*{Comparison with upper bounds in~\Cref{cor:HST}.}
If we set $\sigma$ as a constant and $\Gamma=\ddim$, then the algorithm in~\Cref{cor:HST} achieves distortion $O(\log^2 n \cdot \log\log n)$, the total space $\tilde O(n)$, local space $n^\sigma$, and the round complexity $O(\ddim)$. 
We show that this holds for all metrics with doubling dimension $\ddim \leq O(\sigma \log n)$.
\Cref{thm:LB-hst} implies that this restriction on $\ddim$ is tight. More specifically,~\Cref{thm:LB-hst} implies that 
to achieve $\tilde O(n)$ total space and $2^{o(\ddim)}$ round complexity simultaneously, we require $\ddim\leq c\cdot \sigma \log n$ for some constant $c$.

\subsection{Consistent hashing} 
We prove a lower bound on the number of rounds for sampling a consistent hash and computing its approximate neighborhood hash pairs. 
For a general metric, it requires a polynomial number of rounds in the fully scalable MPC model. 

\begin{theorem}\label{thm:LB-hashing}
Let $\Gamma,\Lambda>1$ and $\sigma,\eps\in(0,1)$ be fixed constants.
Assume that there is an MPC algorithm $\mathcal A$ with access to a distance oracle of an $n$-point metric space of doubling dimension $\ddim$ that samples a $(\Gamma,\Lambda)$-hash and computes the set of approximate neighborhood hash pairs with $O(n^{1+\eps})$ total space and $O(n^\sigma)$ local space.
Then the number of rounds used by $\mathcal A$ is $\Omega(2^\ddim/n^{\eps+\sigma})$.
In particular, for a general metric space where $\ddim=\log n$, the number of rounds used by $\mathcal A$ is $\Omega(n^{1-\sigma-\eps})$.
\end{theorem}
\begin{proof}
    We first show that sampling a $(\Gamma,\Lambda)$-hash and computing approximate hash pairs requires $\Omega(n\cdot 2^\ddim)$ queries. 
    We use the diameter bound $\Delta=\Gamma$.
    Let $c(n) = \Gamma$ for any $n$.
    Consider the distribution $\mathcal D$ over distance functions on a set $P$ of $n$ points defined in the proof of~\Cref{lem:lb-closest}. By~\Cref{lem:lb-closest}, any randomized algorithm that, given a distance function in the support of $\mathcal D$, computes a $c(n)$-approximate closest pair with probability at least $2/3$ requires $\Omega(n\cdot 2^\ddim)$ queries. 
    
    Assume to the contrary that there is a randomized algorithm $\mathcal A$ that samples 
a $(\Gamma,\Lambda)$-hash $\varphi$ and computes its approximate neighborhood hash pairs $\tilde {\mathcal P}_\varphi$
using $o(n\cdot 2^\ddim)$ queries in the worst case.
Let $e=(u,v)$ be the closest pair in $P$. 
We show that, regardless of the random choice made by $\mathcal A$, we can always detect $e$. 
If $\varphi(v)=\varphi(u)$, we immediately find $(u,v)$ by scanning $\varphi(x)$ for all points $x\in P$. More specifically, since the diameter bound is $\Gamma$ and $c(n)=\Gamma$, we have $\varphi(x)\neq \varphi(y)$ for any two distinct points $x,y$ with $x,y\notin e$. 
Therefore, if there is a pair $(x,y)$ with
$\varphi(x)=\varphi(y)$, we can immediately conclude that $(x,y)=e$. Note that scanning $\varphi(\cdot)$ does not make a distance query.
Then consider the case that $\varphi(v)\neq\varphi(u)$. In this case, the set $\tilde {\mathcal P}_\varphi$ of approximate neighborhood hash pairs is $\{(u,\varphi(v)), (v,\varphi(u))\}\cup \{(x,\varphi(x))\mid x\in P\}$ since we set $c=\Delta=\Gamma$; any other distinct point is at distance at least $10\Gamma>2\Delta$, so the locality condition excludes any additional  pairs. Therefore, one can find $e$ by scanning $\tilde{\mathcal P}_\varphi$ without making any distance query.
Therefore, there is a randomized algorithm that always finds $e$ using $o(n\cdot 2^\ddim)$ queries, which contradicts the proof of~\Cref{lem:lb-closest}.

  Now we prove the theorem. 
  Since each machine has at most $n^\sigma$ points and a machine can query only distances between pairs of points currently stored on it, the machine can make at most $n^{2\sigma}$ distance queries in each round. 
    Thus in each round, the number of distance queries made by all the machines is at most $O(n^{2\sigma} \cdot (n^{1+\eps}/n^{\sigma}))=O(n^{1+\sigma+\eps})$.
    Since the total number of queries must be at least $\Omega(n\cdot 2^\ddim)$, 
    the number of rounds is $\Omega(2^{\ddim}/n^{\sigma+\eps})$.  
\end{proof}

\paragraph*{Comparison with upper bounds in~\Cref{thm:main_doubling}.}
If we set $\sigma$ as a constant, then the algorithm in~\Cref{thm:main_doubling} achieves the total space $\tilde O(n)$, the local space $n^\sigma$, and the round complexity $O(\ddim)$. 
We showed in~\Cref{thm:main_doubling} that this holds for all metrics with doubling dimension $\ddim \leq c\cdot \sigma \log n$ for a constant $c$.
\Cref{thm:LB-hashing} implies that this restriction on $\ddim$ is tight. More specifically,~\Cref{thm:LB-hashing} implies that 
to achieve $\tilde O(n)$ total space and $2^{o(\ddim)}$ round complexity simultaneously, we require $\ddim\leq c'\cdot \sigma \log n$ for some constant $c'$.

\section*{Disclosure of AI assistance}
The authors used ChatGPT Astra for discussions that helped inspire the ideas presented in Section~\ref{sec:neighborhood-oracle}. All other research ideas and contributions were independently developed by the authors.
The authors take full responsibility for the content of this paper.

{\small
\bibliographystyle{alpha}
\bibliography{main}
}

\appendix
\section[Comparison with Filtser (2024)]{Comparison with \cite{Fil24scattering}}\label{apx:comparison-fil24}

In this appendix we expand on \Cref{rem:ComparisonFil24}. We will use the notation from this paper to describe \cite{Fil24scattering} as well (the notation changes are cosmetic).
\Cref{lem:seq} follows the same MPX-based framework as \cite[Theorem~10]{Fil24scattering}: choose a net $N$ of centers, assign a random shift $\delta_q$ to every center $q$, and assign each point $x$ to the center maximizing $f_x(q)=\delta_q-\dist_X(x,q)$.
The difference is in the distribution used to sample the shifts $\{\delta_q\}_{q\in N}$, and, more importantly, in the analysis.

\Cref{lem:seq} samples the shifts using (the more common) truncated exponential distribution, where each sample is conditioned to be in $[0,\lambda_\top]$ (recall $\lambda_\top=\Delta/2$).
In contrast, \cite{Fil24scattering} used a \EMPH{betailed exponential distribution}, where one samples from a classical exponential distribution with parameter $\lambda=O(\frac{\ddim}{\Delta})$, and every value above $\lambda_\top$ collapses to $\lambda_\top$. That is, $\Pr[\delta_q=\lambda_\top]=e^{-\lambda\cdot\lambda_\top}$.
Note that the betailed exponential distribution has a point mass (an atom) at the threshold $\lambda_\top$, while the truncated exponential distribution is smooth throughout. 
    
The analysis of \cite{Fil24scattering} proceeds in two steps. 
Fix a point $x$, and indices $s=\tilde{O}(\ddim)$, $m=\tilde{O}(\ddim)\cdot e^{O(\ddim\cdot\frac{r}{\Delta})}$.
First, ignoring the betailing, one orders the centers according to their values $f_x(q)$. It is shown that, with probability at least $1-e^{-\tilde{\Omega}(\ddim)}$, it holds that 
     $f_x(q_{(s)})- f_x(q_{(m)})>2r$.
Second, \cite{Fil24scattering} incorporates the ``betailing''. The probability that $s$ different shifts (among the nearby $N_x$ centers) are betailed (i.e. sampled the value $\lambda_\top$) is at most $e^{-\tilde{O}(\ddim)}$. Accordingly, at least one of $q_{(1)},\dots,q_{(s)}$ will not be betailed, and hence some center among $q_{(1)},\dots,q_{(s)}$ preserves a gap of more than $2r$ from $q_{(m)}$ (w.r.t. the actually used betailed shifts).
As a corollary, there will be at most $m$ intersections.
Finally, the Lov\'asz Local Lemma is used to turn these per-vertex high-probability events into a worst-case guarantee.

The analysis of \cite{Fil24scattering} is a global argument, based on the fact that it is highly unlikely that many of the shifts of the nearby centers are simultaneously betailed.
 In contrast, \Cref{lem:seq} uses a smooth, straightforward argument to bound the expected number of intersections, and then incorporates truncation using simple conditional expectation.
 In fact, \textbf{\Cref{lem:seq} and its analysis are arguably superior}, as they can be used to prove \cite{Fil24scattering} in a simpler and more straightforward way (without relying on the more complicated global argument), and with better hidden constants.
 Indeed, following the proof of \Cref{lem:seq}, recall that $\Lambda$ stands for the number of additional clusters intersecting $B_x$. 
By \cref{eq:tail} it holds that 
\[
\Pr\left[\Lambda\ge m=\tilde{O}(\ddim)\cdot e^{2r\lambda}\right]\le\left(1-e^{-2r\lambda}\right)^{e^{2r\lambda}\cdot\tilde{O}(\ddim)}\le e^{-\tilde{O}(\ddim)}~.
\]
 Then, similarly to \cref{eq:ExpectedIntersections}, we incorporate the truncation and obtain 
 \[
\Pr\left[\Lambda\ge m=\tilde{O}(\ddim)\cdot e^{2r\lambda}\mid\bar{\Psi}\right]\le\frac{\Pr\left[\Lambda\ge m\right]}{1-1/2}\le e^{-\tilde{O}(\ddim)}~.
\]
Thus the simplified analysis yields the desired high-probability bound directly, after which the same Lov\'asz Local Lemma argument as in \cite{Fil24scattering} gives the worst-case sparse-partition guarantee. 
This is the sense in which \Cref{lem:seq} and its analysis are arguably superior: they provide the expected VOP guarantee needed here, while also giving a simpler proof route (with better hidden constants) to the earlier worst-case guarantee.

\section{MPC implementation details for EMD}\label{apd:EMD}

In this section, we show how to obtain an $O(\Gamma\log\Gamma\cdot\log n)$-factor approximation for the Earth Mover's Distance in $O(\ddim/\sigma^3)$ MPC rounds in a metric space of doubling dimension at most $O(\sigma\cdot \log n)$, for every $\Gamma\ge16$, where $n^\sigma$ is the local space complexity. 
Let $(X,\dist)$ be a metric space of doubling dimension $\ddim$. 
For any two finite point sets $A,B\subseteq X$ with $|A|=|B|=n$, the 
\EMPH{Earth Mover's Distance} between $A$ and $B$ 
is defined as 
$\textsf{EMD}(A,B)=\min_{M}\sum_{(a,b)\in M} \dist(a,b)$
where the minimum is taken over all one-to-one matchings between $A$ and $B$.

\paragraph*{Greedy algorithm for EMD.}
Let $T$ be the random HST returned by the MPC algorithm in~\Cref{cor:HST} applied to the metric induced by $A\cup B$. 
We can view $T$ as a rooted labeled tree. 
Then, for every pair $(u,v)$ of points of $A\cup B$,
the expected distortion of $\dist(u,v)$ under $T$ is at most $\alpha=O(\Gamma\log\Gamma\cdot\log n)$, which is $O(\log n)$ for constant $\Gamma$.
The following greedy algorithm computes an optimal matching in the tree metric $T$:
Process the nodes of $T$ in a bottom-up fashion.
For each node $v$, after all its children have been processed, match as many currently unmatched points of $A$ and $B$ inside the subtree rooted at $v$ as possible, and pass the remaining unmatched excess to the parent.
The cost of this algorithm can be computed in the fully-scalable MPC model as follows.

\paragraph*{MPC implementation.}
Let $n^\sigma$ be the local memory bound, for a fixed constant $\sigma\in(0,1)$.
For each node $v$, let $L(v)$ denote its descendant leaves and let $b(v)=|A\cap L(v)|-|B\cap L(v)|$.
If $\ell(v)$ is the distance label of $v$, with zero labels at the leaves, realize $T$ as a weighted tree with edge weights $w(v,\operatorname{par}(v))=(\ell(\operatorname{par}(v))-\ell(v))/2$.
Then, writing $r$ for the root,
\[
\textsf{EMD}_T(A,B)=\sum_{v\ne r}w(v,\operatorname{par}(v))\,|b(v)|.
\]
Indeed, every matching crosses the edge from $v$ to its parent at least $|b(v)|$ times, and the bottom-up greedy algorithm attains equality for every edge.
After applying \Cref{cor:HST} to construct a HST, each point has an
ordered list of $O(\log n)$ labels, one for each level. 
Then using these labels, we can compute 
$\textsf{EMD}_T(A,B)$ in $O(1)$ MPC rounds.

To relate this cost to the original metric, fix an optimal matching $M^*$
between $A$ and $B$ under $\dist$.
Since $T$ dominates $\dist$, every matching has at least its original cost
in $T$, and hence $\textsf{EMD}_T(A,B)\ge\textsf{EMD}(A,B)$.
On the other hand, $M^*$ is also a feasible matching in $T$, so
\[
\mathbb E[\textsf{EMD}_T(A,B)]
\le \sum_{(a,b)\in M^*}\mathbb E[\dist_T(a,b)]
\le \alpha\,\textsf{EMD}(A,B),
\]
where $\dist_T$ denotes distance in $T$.
We repeat this $O(\log n)$ times in parallel, and return the smallest value. 
This proves \Cref{cor:emd}.

\paragraph*{Lower bound for general metrics.}
In a general metric, the situation is significantly more challenging.
More specifically, \cite{azarmehr2025tight} showed that, for every $\delta>0$, there exists a constant $\eps_\delta>0$ such that estimating the unnormalized EMD cost of two disjoint $n$-point sets in a $(1,2)$-metric within an $\eps_\delta\cdot n$-additive error requires $\Omega(n^{2-\delta})$ distance queries.
Here the unnormalized cost is $n$ times the EMD of the uniform probability distributions on these sets, and lies in $[n,2n]$.
For any MPC algorithm with local memory $n^\sigma$ and total memory $O(n^{1+\eps})$, for a constant $\eps>0$ with $\eps+\sigma<1$, it makes at most $O(n^{1+\eps+\sigma})$ distance queries between points stored on the same machine in each round.
Therefore, for $0<\delta<1-\eps-\sigma$, estimating the EMD cost within a \emph{multiplicative error} of $(1+\eps_\delta/2)$ requires $\Omega(n^{1-\eps-\sigma-\delta})=n^{\Omega(1)}$ rounds, since the EMD cost is at most $2n$ in this case.
This shows a sharp contrast between a constant-dimensional Euclidean space and a general metric space.

\end{document}